\documentclass[11pt,draftclsnofoot,onecolumn]{IEEEtran}
\usepackage{bbm}
\usepackage{cite}
\usepackage{amsfonts}
\usepackage{mathrsfs}
\usepackage{graphicx}
\usepackage{amssymb}
\usepackage{latexsym}
\usepackage{amsmath}
\usepackage{stfloats}
\usepackage{cases}
\usepackage{setspace}
\usepackage{bigstrut}
\usepackage{bm}

\usepackage{color,psfrag}

\usepackage{citesort}
\usepackage{comment}
\usepackage{accents}
\usepackage{times,relsize}
\usepackage{latexsym}

\newtheorem{theorem}{Theorem}

\newtheorem{corollary}{Corollary}

\newcommand{\asskip}{\abovedisplayshortskip}
\newcommand{\bsskip}{\belowdisplayshortskip}
\newcommand{\alskip}{\abovedisplayskip}
\newcommand{\blskip}{\belowdisplayskip}

\begin{document}

\title{DFT-Based Hybrid Beamforming Multiuser Systems: Rate Analysis and Beam Selection}

\author{Yu~Han$^\ast$, Shi~Jin$^\ast$, Jun~Zhang$^\dagger$, Jiayi~Zhang$^\ddagger$, and~Kai-Kit~Wong$^\S$\\
$^\ast$National Mobile Communications Research Laboratory, Southeast University, Nanjing, China\\
$^\dagger$Nanjing University of Posts and Telecommunications, Nanjing, China\\
$^\ddagger$Beijing Jiaotong University, Beijing, China\\
$^\S$University College London, London, United Kingdom}
\maketitle

\asskip=3pt
\bsskip=3pt
\alskip=4.3pt
\blskip=4.3pt

\begin{abstract}
This paper considers the discrete Fourier transform (DFT) based hybrid beamforming multiuser system and studies the use of analog beam selection schemes. We first analyze the uplink ergodic achievable rates of the zero-forcing (ZF) receiver and the maximum-ratio combining (MRC) receiver under Ricean fading conditions. We then examine the downlink ergodic achievable rates for the ZF and maximum-ratio transmitting (MRT) precoders. The long-term and short-term normalization methods are introduced, which utilize long-term and instantaneous channel state information (CSI) to implement the downlink power normalization, respectively. Also, approximations and asymptotic expressions of both the uplink and downlink rates are obtained, which facilitate the analog beam selection solutions to maximize the achievable rates. An exhaustive search provides the optimal results but to reduce the time-consumption, we resort to the derived rate limits and propose the second selection scheme based on the projected power of the line-of-sight (LoS) paths. We then combine the advantages of the two schemes and propose a two-step scheme that achieves near optimal performances with much less time-consumption than exhaustive search. Numerical results confirm the analytical results of the ergodic achievable rate and reveal the effectiveness of the proposed two-step method.
\end{abstract}


\section{Introduction}
Massive multiple-input multiple-output (MIMO) technology is a key enabler for the enormous data rate required by the fifth-generation mobile communication \cite{Larsson2014,Andrews2014}. Using a large-scale antenna array, base stations (BSs) can now obtain highly selective beams to pinpoint users \cite{Nam2013,Li2016}. The effects of uncorrelated noise and fast fading are also known to vanish when the number of antennas grows without limit \cite{Marzetta2010}. Massive MIMO provides abundant spatial degrees of freedom for diversity and multiplexing \cite{Adhikary2013,Han2017}. In the early days, the studies of massive MIMO largely focused on the full-digital system where all the signal processing is done at the baseband and each antenna element requires one distinct radio frequency (RF) chain. The large number of expensive transceivers and the huge amount of power consumption nevertheless become the bottlenecks that limit the developments of massive MIMO systems.

To overcome these problems, low-cost solutions have been proposed, ranging from, for example, decreasing the number of RF chains \cite{Huang2010,Bogale2016} to lowering the resolutions of analog-to-digital converters \cite{Fan2015,Zhang2016,Zhang2017}. One such example is the hybrid beamforming architecture, which uses a small number of RF chains to control the large-scale antenna array. Different from the full-digital system, there are two beamforming components in the hybrid beamforming system. One is the high-dimensional analog beamforming implemented at the RF module and another is the low-dimensional digital beamforming implemented at the baseband module. Due to the non-linear characteristics of power amplifiers, it is not suggested to adjust the amplitude of the signals for beamforming use at the RF module. The commonly used analog beamforming enablers include phase shifters \cite{Sohrabi2016,Han2015,Alkhateeb2015}, switch networks \cite{Mendez2015}, lens antennas \cite{Brady2013,Zeng2016}, and Butler matrices or other discrete Fourier transform (DFT) modules embedded on field-programmable analog arrays \cite{Molisch2003,Suh2011}. All the above devices only shift the phase of the signal without changing its modulus. Due to the constant-modulus restriction at the analog component, it is important to design proper analog beamforming weights for best performance. The methods to design the analog beamforming weights can be classified into two categories. One is the non-analog-codebook based design, where the analog beamforming weights are first calculated from a closed-form expression and then regulated according to the hardware constraint \cite{Sohrabi2016,Park2017}. For example, \cite{Park2017} computed the analog beamforming weights based on the spatial channel matrix and then iteratively adjusted the weights to satisfy the constraint of constant-module. The other is the analog-codebook based design, where an analog codebook that contains more than one beam is predefined and the analog beamforming weights are selected from the codebook.

In the analog-codebook based hybrid beamforming design, many efforts were paid on analog beam selection. In particular, \cite{Liu2014} adopted the DFT codebooks at the analog beamforming module and formulated the hybrid beamforming design into an optimization problem. Also, \cite{Alkhateeb2014} and \cite{Noh2015} introduced the multi-stage or multi-resolution codebook which would allow hierarchical searching and could significantly reduce the time for beam selection. Unfortunately, the selection process cannot be performed by more than one user in parallel, which hinders the application of multi-stage codebooks. In addition, in \cite{El2014}, it was proposed to decompose the full-digital beamforming weights into two parts: the analog part and the digital part by an orthogonal matching pursuit algorithm. To be more specific, the analog beamforming weights were selected from a vector set. Each vector in the set was a steering vector of the antenna array pointing to a sampled spatial direction. The weight decomposition based algorithm performed well in single-user systems, but suffered from inter-user interference if applied in the multiuser case. Practical analog beam selection schemes for multiuser systems are therefore needed.

All of the above-mentioned designs are based on the phase shifter networks which suffer from problems, such as difficult-to-implement integration and high energy consumption, etc. In light of this, low-cost and easy-to-implement devices such as Butler matrices have since gained much importance \cite{Butler1961}. Results in \cite{Garcia2016} demonstrate that the Butler matrix based DFT analog beamforming network introduces less power losses and outperforms the phase shifter based fully connected analog beamforming network with more RF chains. Later, \cite{Tan2017} analyzed the achievable rate of the DFT-based multiuser hybrid system in Rayleigh fading channels when the zero-forcing (ZF) receiver was employed in the uplink. However, the DFT beams in \cite{Tan2017} were fixed and could not be switched according to the actual situations. For this reason, it has motivated the use of the analog beam selection schemes of the Butler matrix based hybrid beamforming for multiuser systems.

This paper focuses on the Butler matrix based hybrid beamforming architecture and investigates analog beam selection schemes for multiuser systems. We use Ricean fading channel models to account for typical current and future applications, for example, machine-type communications \cite{Zhang2017}. In order to find the beams that optimize the rate performance, we first derive the approximations of the ergodic achievable rate\footnote{For convenience, the terms ``achievable rate'' and ``ergodic achievable rate'' are used interchangeably for the rest of this paper.} under the assumption that the analog beams are fixed. Then based on the ergodic rate analysis, we obtain optimal and suboptimal beam selection solutions. Note that the selection results are effective during the channel's coherent time. The main contributions of this paper are twofold:
\begin{itemize}
\item {\em Approximations of the uplink and downlink achievable rates}---We first analyze the uplink achievable rates of the ZF and the maximum-ratio combining (MRC) receivers and derive their approximations, respectively. The effect of the analog beamformed line-of-sight (LoS) paths on the achievable rate is examined. We demonstrate that the orthogonality of the analog beamformed LoS paths from different users as well as the the complete projection of the LoS paths on the selected beams contribute to high rates. Then utilizing a similar approach for the downlink, we obtain approximations of the achievable rate for the ZF and the maximum-ratio transmitting (MRT) precoders when the long-term normalization and short-term normalization methods are adopted, respectively. We find that for ZF precoders short-term normalization gives higher rate than the long-term normalization case, while in the case of MRT precoders the result is not conclusive.
\item {\em Practical analog beam selection schemes}---The approximations and asymptotic expressions of the achievable rates help us develop more efficient DFT beam selection schemes. An approximation-based exhaustive search is first introduced to achieve the optimum performance at the price of the highest time-consumption. In particular, referring to the observation on the effect of the projected power of the LoS paths on the rate, we propose a projected power based per-user selection scheme to choose beams according to the maximum projected power on the LoS paths for each user. The scheme reduces the computation time greatly, but ignores the inter-user interference. To tackle this, we propose a two-step selection scheme where we perform the per-user selection in the first step assuming there are more RF chains than they have before the extra beams are removed using the asymptotic rate expressions in the second step. This scheme strikes the balance between performance and time-consumption.
\end{itemize}
It is worth emphasizing that a distinguishing feature from the prior work in \cite{Tan2017} is that our work aims to design analog beam selection schemes for Ricean fading channels. The proposed beam selection schemes can be applied to uplink ZF/MRC receivers, as well as downlink ZF/MRT precoders with both long-term and short-term normalization methods.

The rest of the paper is organized as follows. Section \ref{Sec:System Model} introduces the hybrid beamforming multiuser system, including the Butler matrix based hybrid architecture, Ricean fading channel, and uplink and downlink signal models. Sections \ref{Sec:Uplink Rate} and \ref{Sec:Downlink Rate} analyze the uplink achievable rates of the ZF/MRC receivers and the downlink rates of the ZF/MRT precoders, respectively. Section \ref{Sec:Beam Selection} presents and compares the three analog beam selection schemes. The numerical results are shown in Section \ref{Sec:Numerical Results}. Finally, Section \ref{Sec:Conclusion} concludes the paper.

\emph{Notations}---In this paper, matrices and vectors are denoted by uppercase and lowercase boldface letters, respectively. We use ${\mathbf{I}}$ to represent the identity matrix. The superscripts $(\cdot)^\dag$, $(\cdot)^{H}$, $(\cdot)^{T}$, and $(\cdot)^{*}$ denote, respectively, the pseudo-inverse, conjugate-transpose, transpose, and conjugate operations. $\mathbb{E}\{\cdot\}$ represents the expectation with respect to all random variables within the brackets. We also use $\left| \cdot \right|$ and $\left\| \cdot \right\|$ to denote taking absolute value and modulus operations respectively, and $\left\lfloor \cdot \right\rfloor$ to represent rounding a decimal to its nearest lower integer.

\section{System Model}\label{Sec:System Model}
Consider a massive MIMO multiuser system where the BS is located at the cell center and communicates with $N_u$ single-antenna users on the same time-frequency resource block.

\subsection{Hybrid Beamforming}

\begin{figure}
  \centering
  \includegraphics[scale=0.45]{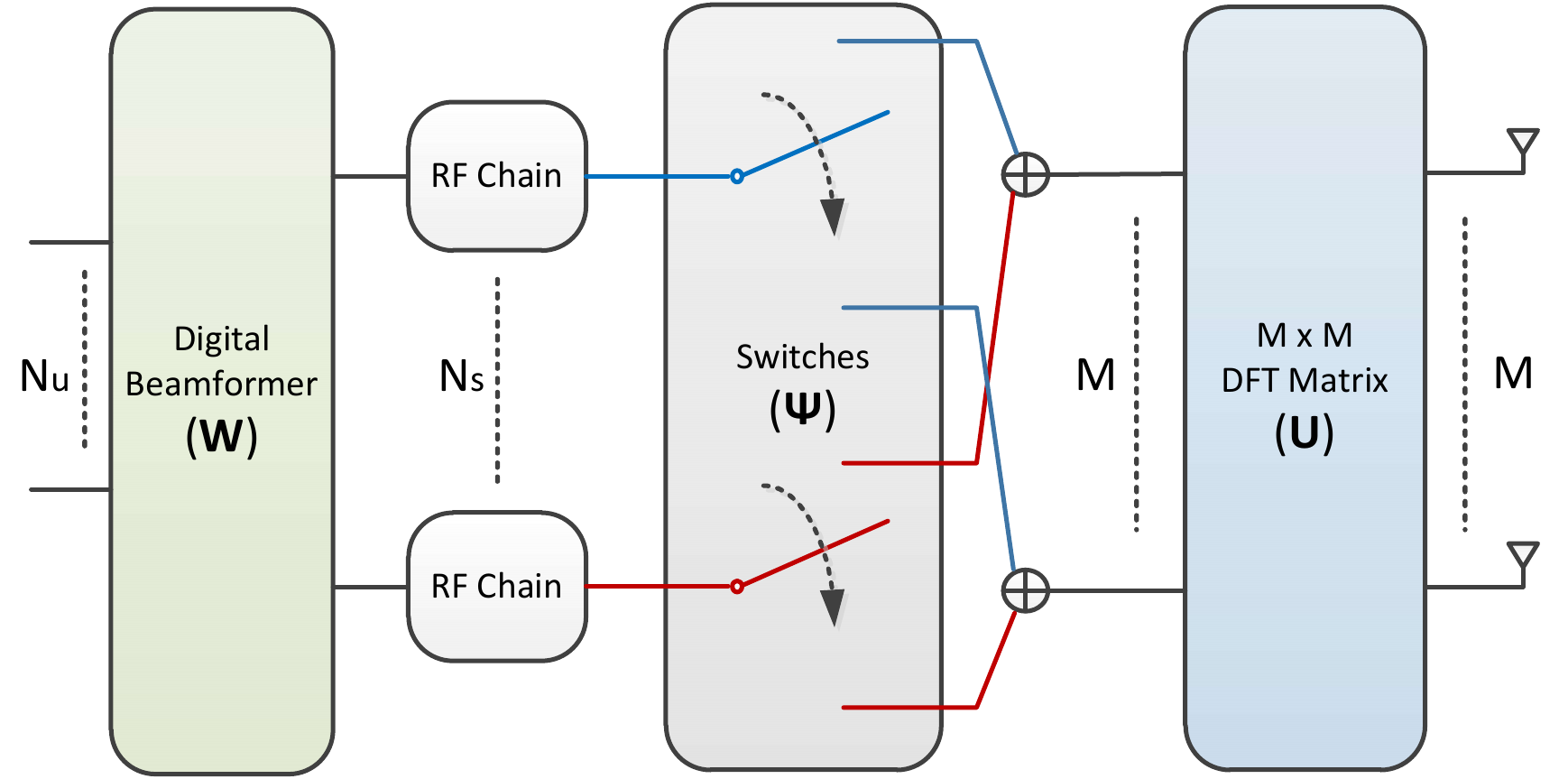}
  \caption{The BS adopts the Butler Matrix based hybrid architecture. Each RF chain connects with one of the $M$ ports by controlling its own switch.}\label{ButlerMatrix}
\end{figure}

In this paper, we focus on the low-cost hybrid analog-and-digital architecture. The Butler matrix based hybrid architecture is adopted at the BS as shown in Fig.\ref{ButlerMatrix}. The total number of BS antenna elements is $M$. The $M$ dimensional DFT matrix corresponds to $M$ DFT beam ports. We assume that there are $N_s$ ($N_u \le N_s \le M$) RF chains. Each RF chain controls its switch to connect with one of the $M$ beam ports.

In this model, we denote the $M$ dimensional DFT matrix as $\bf{U}$, which is written as
\begin{equation}\label{DFT codebook}
{\bf{U}} = \frac{1}{\sqrt{M}} \left[
\begin{matrix}
1 & 1 & \cdots & 1\\
1 & e^{j2\pi\frac{1}{M}} & \cdots & e^{j2\pi\frac{M-1}{M}} \\
\vdots & \vdots & \quad & \vdots\\
1 & e^{j2\pi\frac{M-1}{M}} & \cdots & e^{j2\pi\frac{(M-1)^2}{M}}
\end{matrix}
\right].
\end{equation}
The analog beamforming matrix is constructed with $N_s$ rows of $\bf{U}$, i.e.,
\begin{equation}\label{ABF matrix}
\bf{F} = \bf{\Psi} \bf{U},
\end{equation}
where ${\bf{\Psi}} = \left[ {\bf{e}}_{i_1}, {\bf{e}}_{i_2},\dots, {\bf{e}}_{i_{N_s}} \right]^T$ and ${\bf{e}}_j \in \mathbb{Z}^{M \times 1}$ is a vector with the $j$th element being $1$ and zero elsewhere. Since $\bf{U}$ is invariable, the beam selection matrix $\bf{\Psi}$ plays a decisive role in the analog beamforming. In this paper, we will investigate the effect of $\bf{\Psi}$ on the achievable rate and exploit this in the beam selection schemes in the subsequent sections. In order to simplify the expressions in the following analysis, we adopt the same notation ${\bf{F}}$ to represent the uplink and the downlink analog beamformer. It should be noted that ${\bf{F}}$ may consist of different DFT beams in uplink and downlink.

\subsection{Channel Model}
To fully describe the characteristics of the wireless channel, we write the multiuser MIMO channel as
\begin{equation}\label{multiuser channel G}
{\bf G} = {\bf{HD}}^{\frac{1}{2}},
\end{equation}
where ${\bf{D}} = {\rm diag}\{ \beta_1,\beta_2,\dots,\beta_{N_u} \}$, in which $\beta_k \in \mathbb{R}^+$ reflects the energy of the $k$th user channel, and ${\bf{H}} \in \mathbb{C}^{M \times N_u}$ denotes the fast fading factor matrix which models the propagation condition of the channel. Here, we focus on the Ricean fading condition, so that $\bf{H}$ is written as \cite{Zhang2017}
\begin{equation}\label{multiuser channel H}
{\bf{H}} = \bar{\bf{H}} \left[ {\bf{\Omega}} \left( {\bf{\Omega}}+{\bf{I}}_{N_u} \right)^{-1} \right]^{\frac{1}{2}}
+ {\bf{H}}_w \left[ \left( {\bf{\Omega}}+{\bf{I}}_{N_u} \right)^{-1} \right]^{\frac{1}{2}},
\end{equation}
where ${\bf{\Omega}} = {\rm diag} \left( K_1, K_2,\dots, K_{N_u} \right)$, in which $K_k$ is the Ricean $K$-factor of the $k$th user channel, $\bar{\bf{H}}\in \mathbb{C}^{M \times N_u}$ is the deterministic LoS component with the $k$th column ${\bar{\bf{h}}}_k$ referring to user $k$, ${\bf{H}}_w \in \mathbb{C}^{M \times N_u}$ denotes the random component with independent and identically distributed (i.i.d.) elements, and each element of ${\bf{H}}_w$ is a complex Gaussian random variable with zero mean and unit variance.

\subsection{Signal Model}
In the uplink, the BS receives the signals from all the $N_u$ users. We assume that each user has equal transmit power. Then, the uplink received signal vector at the BS antennas can be written as
\begin{equation}\label{uplink received signal}
{\bf{r}} = \sqrt {{P_{avg}}} {\bf{Gs}} + {\bf{n}},
\end{equation}
where ${P_{avg}}$ is the transmit power of each user, ${\bf{G}}$ is defined in \eqref{multiuser channel G}, ${\bf{s}} \in \mathbb{C}^{{N_u} \times 1}$ is the transmit signal vector satisfying $\mathbb{E} \left\{ {{\bf{s}}{{\bf{s}}^H}} \right\} = {{\bf{I}}_{{N_u}}}$, ${\bf{n}} \in \mathbb{C}^{M \times 1}$ denotes the complex Gaussian noise vector, and each element of ${\bf{n}}$ has zero mean and unit variance. Due to the hybrid beamforming structure at the BS, the received signals will be firstly analog beamformed via the Butler matrix network, and then equalized via the digital beamformer. Therefore, $\bf{r}$ is further processed by
\begin{equation}\label{uplink processed signal}
{\bf{y}} = {{\bf W}_U}{{\bf F}} {\bf{r}} = \sqrt {{P_{avg}}} {{\bf W}_U} {{\bf F}} {\bf{Gs}} + {{\bf W}_U}{{\bf F}}{\bf{n}},
\end{equation}
where ${{\bf W}_U} \in \mathbb{C}^{{N_u} \times {N_s}}$ is the digital beamformer which is considered as the MIMO receiver in the uplink. Then we use the beamformed signal $\bf{y}$ to estimate the original signal $\bf{s}$. It should be noted that the noise is also beamformed. Let us analyze the components of $\bf{y}$ and write the $k$th data stream as
\begin{equation}\label{uplink processed signal stream k}
y_k = \sqrt {P_{avg}} {\bf{w}}_{U,k}{{\bf F}}{{\bf{g}}_{k}}{s_k} + \sum\limits_{j \ne k} {\sqrt {P_{avg}} {\bf{w}}_{U,k}{{\bf F}}{{\bf{g}}_{j}}{s_j}}  + {\bf{w}}_{U,k}{{\bf F}}{\bf{n}},
\end{equation}
where ${\bf{w}}_{U,k}$ and ${\bf{g}}_{k}$ are the $k$th row and column vectors of ${\bf W}_U$ and $\bf{G}$, respectively. We can easily find that in addition to the colored noise, inter-user interference exists as well.

Similarly, in the downlink, the BS transmits the hybrid beamformed signals to all users simultaneously. The received signal vector at the user side is expressed as
\begin{equation}\label{downlink received signal}
{\bf{r}} = \sqrt P {\bf{G}}^T {{\bf F}}^T {{\bf W}_D} {\bf{x}} + {\bf{n}},
\end{equation}
where $P$ is the total transmit power at the BS, ${{\bf W}_D} \in \mathbb{C}^{{N_s} \times {N_u}}$ is the downlink digital precoder satisfying ${\left\| {{\bf W}_D} \right\|_F}{\rm{ = 1}}$, ${\bf{x}} \in \mathbb{C}^{{N_u} \times 1}$ is the transmit signal vector satisfying $\mathbb{E} \left\{ {{\bf{x}}{{\bf{x}}^H}} \right\} = {{\bf{I}}_{{N_u}}}$, and ${\bf{n}} \in \mathbb{C}^{{N_u} \times 1}$ is the complex Gaussian noise vector with each element having zero mean and unit variance. For the $k$th user, the received signal also contains the target signal, the interference and the noise, i.e.,
\begin{equation}\label{downlink received signal at user k}
r_{k} = \sqrt P {\bf{g}}_k^T{{\bf F}}^T{{\bf{w}}_{D,k}}{x_k} + \sum\limits_{j \ne k} {\sqrt P {\bf{g}}_k^T{{\bf F}}^T{{\bf{w}}_{D,j}}{x_j}}  + {n_{k}},
\end{equation}
where ${\bf{w}}_{D,k}$ is the $k$th column vector of ${{\bf W}_D}$.

\section{Uplink Rate Analysis}\label{Sec:Uplink Rate}
To perform analog beam selection for the DFT-based hybrid beamforming system, we first choose proper digital beamformers and analyze their performance assuming that the analog beamformer $\bf{F}$ is fixed. Considering the reduced dimensional processing at the baseband module, conventional MIMO techniques are applicable in the digital beamforming design. In this section, we focus on the uplink and evaluate two popular linear receivers, namely the ZF and MRC receivers, by deriving the approximations of the achievable rates.

\subsection{ZF Receiver}
If the ZF receiver is used in the uplink, ${\bf W}_U$ is written as
\begin{equation}\label{ZF receiver}
{\bf{W}}_U = \sqrt {\frac{1}{{{P_{avg}}}}} {\bf{G}}_{eq}^\dag  = \sqrt {\frac{1}{{{P_{avg}}}}} {\left( {{\bf{G}}_{eq}^H{{\bf{G}}_{eq}}} \right)^{ - 1}}{\bf{G}}_{eq}^H,
\end{equation}
where
\begin{equation}\label{effective uplink channel}
{{\bf{G}}_{eq}} = {\bf{FG}}
\end{equation}
is the effective channel seen from the air interface. Then, \eqref{uplink processed signal} is expressed as
\begin{equation}\label{ZF uplink signal}
{\bf{y}} = {\bf{s}} + \sqrt {\frac{1}{{{P_{avg}}}}} {\bf{G}}_{eq}^\dag {\bf{Fn}},
\end{equation}
which will be used for detection. Since ${\bf{F}}$ is abstracted from the DFT matrix and ${\bf{n}}$ is a multivariable Gaussian random vector, ${\bf{Fn}}$ can be seen as a new $N_s$-dimensional complex Gaussian vector with each element having zero mean and unit variance. Hence, the achievable rate is calculated as
\begin{equation}\label{ZF UL rate real1}
{R^{\rm{ZF}}} = \sum\limits_{k = 1}^{{N_u}} \mathbb{E}{\left\{ {{{\log }_2}\left( {1 + \frac{{{P_{avg}}}}{{{{\left[ {{{\left( {{\bf{G}}_{eq}^H{{\bf{G}}_{eq}}} \right)}^{ - 1}}} \right]}_{k,k}}}}} \right)} \right\}},
\end{equation}
where $\left[{\bf{A}}\right]_{m,n}$ represents the $(m,n)$th entry of matrix $\bf{A}$.

Based on \eqref{ZF UL rate real1}, we provide the approximation of the achievable rate of the ZF receiver and have the following theorem.

\begin{theorem}\label{UL ZF appr1}
When the ZF receiver is adopted in the uplink of the DFT-based hybrid beamforming system, the achievable rate can be approximated as
\begin{equation}\label{ZF UL rate appr1}
R_{{\rm{App}}}^{\rm{ZF}} \!=\! \sum\limits_{k = 1}^{{N_u}} {{{\log }_2}\left( {1 \!+\! {P_{avg}}{\beta _k}{\varepsilon _k} \exp \left( \psi \left( {{N_s} - {N_u} + 1} \right) \right)} \right)},
\end{equation}
where $\psi \left(  \cdot  \right)$ denotes the digamma function,
\begin{equation}\label{ZF UL rate appr1 epsk}
{\varepsilon _k} = \frac{\prod\limits_{i = 1}^{{N_u}} {{\alpha _i}}}{\prod\limits_{i = 1}^{{N_u} - 1} {{{\bar \alpha }_{k,i}}}},
\end{equation}
${\left\{ {{\alpha _i}} \right\}_{i = 1,\dots,{N_u}}}$ are the eigenvalues of
\begin{equation}\label{central Wishart covariance}
{\bf{\hat \Sigma }} = {\left( {{\bf{\Omega }} + {{\bf{I}}_{{N_u}}}} \right)^{ - 1}} + \frac{1}{{{N_s}}}{{\bf{ T}}^H}{\bf{T}}
\end{equation}
sorted in an ascending order, ${\bf{T}} = {\bf{F\bar H}}{\left[ {{\bf{\Omega }}{{\left( {{\bf{\Omega }} + {{\bf{I}}_{{N_u}}}} \right)}^{ - 1}}} \right]^{\frac{1}{2}}}$, ${\left\{ {{{\bar \alpha }_{k,i}}} \right\}_{i = 1,\dots,{N_u} - 1}}$ are the eigenvalues of
\begin{equation}\label{covariance user k}
{{{\bf{\hat \Sigma }}}_k} = {\left( {{{{\bf{\bar \Omega }}}_k} + {{\bf{I}}_{{N_u} - 1}}} \right)^{ - 1}} + \frac{1}{{{N_s}}}{\bf{T}}_k^H{{\bf{T}}_k}
\end{equation}
in an increasing order, ${{\bf{T}}_k} = {\bf{F}}{{{\bf{\bar H}}}_k}{\left[ {{{{\bf{\bar \Omega }}}_k}{{\left( {{{{\bf{\bar \Omega }}}_k} + {{\bf{I}}_{{N_u} - 1}}} \right)}^{ - 1}}} \right]^{\frac{1}{2}}}$, ${{{\bf{\bar H}}}_k}$ is ${\bf{\bar H}}$ with the $k$th column removed, and ${{\bf{\bar \Omega }}_k}$ is ${\bf{\Omega}}$ with the $k$th row and the $k$th column removed.
\end{theorem}

\begin{proof}
See Appendix \ref{Apdx:TheoULZF1}.
\end{proof}

From \eqref{ZF UL rate appr1}, we can find that when $N_s$, $N_u$ and $\beta_k$ are fixed, the rate of the ZF receiver increases in proportional to ${\varepsilon _k}$, which is decided by the selected analog beams and the Ricean components of the user channels.
For better understanding of the effect of ${\varepsilon _k}$ on the achievable rate, we give insights on some special cases. We start with the asymptotic analysis of the achievable rate under Rayleigh fading conditions.

\begin{corollary}\label{ZF UL appr1 cor1}
When Ricean fading reduces to Rayleigh fading, i.e., $K_k = 0$ for $k = 1,\dots,N_u$, the approximation \eqref{ZF UL rate appr1} is rewritten as
\begin{equation}\label{ZF UL rate cor1}
R_{{\rm{App}}}^{\rm{ZF}} = \sum\limits_{k = 1}^{{N_u}} {{{\log }_2}\left( {1 + {P_{avg}}{\beta _k} \exp \left( \psi \left( {{N_s} - {N_u} + 1} \right) \right)} \right)}.
\end{equation}
\end{corollary}

\begin{proof}
In this case, ${\bf{\hat \Sigma }}$ and ${{{\bf{\hat \Sigma }}}_k}$ are reduced to identity matrices. As such, their eigenvalues satisfy ${\alpha _k}=1 $ for $k = 1,\dots,{N_u}$ and ${\bar \alpha }_{k,i}=1$ for $k = 1,\dots,{N_u}, i = 1,\dots,{N_u} - 1$. Hence, $\varepsilon _k = 1$ and we get the desired results.
\end{proof}

With Rayleigh fading, LoS paths no longer exist, and $\varepsilon _k$ has no effect on the achievable rate. The beam selection results make no difference to the system performance. Thus, we can choose arbitrary beams. Besides, the achievable rate can be enhanced by increasing the number of RF chains. Note that if we set $N_s = M$, then \eqref{ZF UL rate cor1} coincides with \emph{Proposition 2} of \cite{Matthaiou2011}, which illustrates the effectiveness of \emph{Theorem \ref{UL ZF appr1}}.

We regard the Rayleigh fading case as a reference and now investigate ${\varepsilon _k}$ in Ricean fading environments. Obviously, ${\bf{\hat \Sigma }}$ is Hermitian and ${{{\bf{\hat \Sigma }}}_k}$ can be seen as ${\bf{\hat \Sigma }}$ with the $k$th row and the $k$th column removed. According to \emph{Theorem 4.3.8} of \cite{Horn1985}, the eigenvalues of ${\bf{\hat \Sigma }}$ and ${{{\bf{\hat \Sigma }}}_k}$ satisfy
\begin{equation}\label{eigenvalues comparison}
\alpha_1 \le {\bar \alpha }_{k,1} \le \alpha_2 \le \cdots \le {\bar \alpha }_{k,N_u-1} \le \alpha_{N_u}.
\end{equation}

We now move on to another special case of Ricean fading and derive the asymptotic rate in the following corollary.

\begin{corollary}\label{ZF UL appr1 cor2}
In the case that $K_k \to \infty $ for $k=1,\dots,N_u$, if the effective LoS components hold orthogonality, i.e., ${{\bf{\bar h}}_j^H{{\bf{F}}^H} {\bf{F}}{{{\bf{\bar h}}}_k}}=0$ for $j \ne k$, then \eqref{ZF UL rate appr1} approaches to
\begin{equation}\label{ZF UL rate cor2 approach}
R_{{\rm{App}}}^{\rm{ZF}} \to \sum\limits_{k = 1}^{{N_u}} {R_{{\rm{App}},k}^{\rm{ZF}}} ,
\end{equation}
where
\begin{equation}\label{ZF UL rate cor2}
{R_{{\rm{App}},k}^{\rm{ZF}}} = {{{\log }_2}\left( {1 \!+\! {P_{avg}}{\beta _k}{\hat \varepsilon _k} \exp \left( \psi \left( {{N_s} \!-\! {N_u} \!+\! 1} \right) \right)} \right)},
\end{equation}
and ${\hat \varepsilon _k}={\left\| {{\bf{F\bar h}}_k} \right\|^2}/{N_s}$.
\end{corollary}

\begin{proof}
See Appendix \ref{Apdx:TheoULZF1cor2}.
\end{proof}

We find that ${R_{{\rm{App}},k}^{\rm{ZF}}}$ is in proportion to $\left\| {{\bf{F\bar h}}_k} \right\|^2$ when ${P_{avg}}$, ${\beta _k}$ and $N_s$ are fixed. If the LoS paths are completely projected on the selected beams, then $\left\| {{\bf{F\bar h}}_k} \right\|^2 = M$ and ${\varepsilon _k^{(1)}} \gg 1$ when $M \gg N_s$. On the contrary, if there is very little power projected on the beams, then $\left\| {{\bf{F\bar h}}_k} \right\|^2 \approx 0$ and ${\varepsilon _k^{(1)}} \approx 0$. Therefore, good beam selection results are critical under Ricean fading conditions. It is suggested to select beams that cover the LoS paths and meanwhile contribute to the orthogonality among the effective LoS paths from different users.

\subsection{MRC Receiver}
MRC is another well-known linear receiver, which combines the received signals on multiple RF chains to enhance the signal power. When adopting the MRC receiver in the uplink, ${\bf{W}}_U$ is expressed as
\begin{equation}\label{MRC receiver}
{\bf{W}}_U = {\bf{G}}_{eq}^H,
\end{equation}
and the combined signal vector is
\begin{equation}\label{MRC combined signal vector}
{\bf{y}} = \sqrt {{P_{avg}}} {\bf{G}}_{eq}^H {\bf{G}}_{eq}{\bf{s}} + {\bf{G}}_{eq}^H{\bf{F}}{\bf{n}}.
\end{equation}
Then the achievable rate of the MRC receiver is calculated as
\begin{equation}\label{MRC UL rate real}
{R^{\rm{MRC}}} = \sum\limits_{k = 1}^{{N_u}} \mathbb{E}{\left\{ {{{\log }_2}\left( {1 + {\gamma _k}} \right)} \right\}} ,
\end{equation}
where
\begin{equation}\label{MRC UL rate gammak}
{\gamma _k} = \frac{{{P_{avg}}{{\left\| {{{\bf{g}}_{eq,k}}} \right\|}^4}}}{{{P_{avg}}\sum\limits_{j \ne k} {{{\left| {{\bf{g}}_{eq,k}^H{{\bf{g}}_{eq,j}}} \right|}^2}}  + {{\left\| {{{\bf{g}}_{eq,k}}} \right\|}^2}}}
\end{equation}
reflects the signal to interference-and-noise ratio (SINR), and ${{\bf{g}}_{eq,k}}$ denotes the $k$th column vector of ${{\bf{G}}_{eq}}$. Taking a similar approach as previously, we have the following theorem to evaluate the achievable rate performance.

\begin{theorem}\label{UL MRC appr}
When adopting the MRC receiver in the DFT-based hybrid beamforming system, the uplink achievable rate can be approximated as
\begin{equation}\label{MRC UL rate appr}
R_{{\rm{App}}}^{\rm{MRC}} = \sum\limits_{k = 1}^{{N_u}} {{{\log }_2}\left( {1 + \frac{{\frac{{P_{avg}}{{\beta _k}}}{{ {{K_k} + 1} }}{\chi _1^{(k)}}}} {{\sum\limits_{j \ne k} {\frac{{P_{avg}}{{\beta _j}}}{{{K_j} + 1}}{\chi _{2,j}^{(k)}}} + {\chi _3^{(k)}}}}} \right)},
\end{equation}
where
\begin{align}
{\chi _3^{(k)}}& \triangleq {K_k}{\left\| {{\bf{F}}{{{\bf{\bar h}}}_k}} \right\|^2} + {N_s},\label{MRC UL rate appr x3}\\
{\chi _1^{(k)}}& \triangleq {{\chi _3^{(k)2}}} + 2{\chi _3^{(k)}} - N_s,\label{MRC UL rate appr x1}
\end{align}
and
\begin{equation}\label{MRC UL rate appr x2}
{\chi _{2,j}^{(k)}} \triangleq {{K_k}{K_j}{{\left| {{\bf{\bar h}}_j^H{{\bf{F}}^H} {\bf{F}}{{{\bf{\bar h}}}_k}} \right|}^2} + {\chi _3^{(j)}} + {\chi _3^{(k)}} - N_s}.
\end{equation}
\end{theorem}

\begin{proof}
See Appendix \ref{Apdx:TheoULMRC}.
\end{proof}

To analyze the effect of the analog beamformed LoS paths on the achievable rate, we first derive the asymptotic expression of the achievable rate in Rayleigh fading channels.

\begin{corollary}\label{MRC UL appr cor1}
When $K_k = 0$ for $k = 1,\dots,N_u$, Ricean fading reduces to Rayleigh fading, and \eqref{MRC UL rate appr} is rewritten as
\begin{equation}\label{MRC UL cor1}
R_{{\rm{App}}}^{\rm{MRC}}= \sum\limits_{k = 1}^{{N_u}} {{{\log }_2}\left( 1+ \frac{{P_{avg}}{\beta _k}\left(N_s+1 \right)}{\sum\limits_{j \ne k} {{P_{avg}}{\beta _j}}+1}  \right)}.
\end{equation}
\end{corollary}

\begin{proof}
In Ricean fading channels, ${{\bf{g}}_{eq,k}}$ is expressed as
\begin{equation}\label{MRC UL appr cor1 proof}
{{\bf{g}}_{eq,k}} = \sqrt {\beta _k} {\bf{Fh}}_{w,k}.
\end{equation}
Accordingly, the expectation items can be calculated as
\begin{align}
\mathbb{E}\left\{ {{{\left\| {{{\bf{g}}_{eq,k}}} \right\|}^4}} \right\} &= {\beta _k^2}{\left( {{N_s^2} + N_s} \right)},\label{MRC UL appr cor1 x1}\\
\mathbb{E}\left\{ {{{\left| {{\bf{g}}_{eq,k}^H{{\bf{g}}_{eq,j}}} \right|}^2}} \right\}&={\beta _k}{\beta _j}{N_s},\label{MRC UL appr cor1 x2}\\
\mathbb{E}\left\{ {{{\left\| {{{\bf{g}}_{eq,k}}} \right\|}^2}} \right\} &= {\beta _k}{N_s}.\label{MRC UL appr cor1 x3}
\end{align}
Applying \eqref{MRC UL appr cor1 x1}--\eqref{MRC UL appr cor1 x3} into \eqref{MRC UL rate appr proof1} which is shown in Appendix \ref{Apdx:TheoULMRC}, we can obtain \eqref{MRC UL cor1}, which completes the proof.
\end{proof}

Similarly, with Rayleigh fading, the analog beam selection results make no difference and we can choose arbitrary beams. Note that when it comes to the full-digital system and if large-scale fading coefficients satisfy $\beta_1 = \beta_2 = \cdots = \beta_{N_u} = 1$, the derived expression from \eqref{MRC UL cor1} is exactly in accordance with \cite{Zhang2014}. This verifies the effectiveness of \emph{Theorem \ref{UL MRC appr}}.

The following corollary provides the achievable rate limit in the case of pure Ricean fading.

\begin{corollary}\label{MRC UL appr cor3}
When $K_k \to \infty$ for $k = 1,\dots, N_u$ and the effective LoS components hold orthogonality, the achievable rate of the MRC receiver approaches to
\begin{equation}\label{MRC UL cor3 approach}
R_{{\rm{App}}}^{\rm{MRC}} \to \sum\limits_{k = 1}^{{N_u}} {R_{{\rm{App}},k}^{\rm{MRC}}},
\end{equation}
where
\begin{equation}\label{MRC UL cor3}
{R_{{\rm{App}},k}^{\rm{MRC}}} = {{{\log }_2}\left( 1+{P_{avg}}{\beta _k}{\left\| {{\bf{F\bar h}}_k} \right\|^2}  \right)}.
\end{equation}
\end{corollary}

From \eqref{MRC UL cor3}, we can obtain similar insights with \emph{Corollary \ref{ZF UL appr1 cor2}} that the more the power projected on the selected beams from the LoS paths, the higher the rate one can achieve in Ricean fading. Moreover, if the number of RF chains configured at the BS is increased, the projected power of LoS paths on the selected beams is enhanced as well, which further contributes to the improvement of the achievable rate.

\section{Downlink Rate Analysis}\label{Sec:Downlink Rate}
In this section, we focus on the downlink of the DFT-based hybrid beamforming system. Following the similar approach as in the uplink, we analyze the downlink rates when adopting the ZF and the MRT precoders. Considering the total transmit power constraint at the BS, power normalization is required in the downlink, which differs from the uplink analysis.

\subsection{ZF Precoder}
In the downlink, the ZF precoder is designed as
\begin{equation}\label{ZF precoder}
{\bf{W}}_D = {\bf{\bar W}}{\bf{P}},
\end{equation}
where
\begin{equation}\label{ZF DL matrix}
{\bf{\bar W}} = {{\bf{G}}_{eq}^{*}}{\left( {\bf{G}}_{eq}^T {{\bf{G}}_{eq}^{*}}  \right)^{ - 1}}
\end{equation}
is the zero-forcing matrix,
\begin{equation}\label{power allocation}
{\bf{P}} = {\rm diag}\left( {{\rho _1},{\rho _2},\dots,{\rho _{{N_u}}}} \right)
\end{equation}
is the power normalization matrix, and $\rho _k \in \mathbb{R}^{+}$ is the normalization coefficient for the $k$th data stream. Since ${\bf{G}}_{eq}^T{\bf{\bar W}} = {{\bf{I}}_{{N_u}}}$, the achievable rate of the ZF precoder is
\begin{equation}\label{ZF DL rate real}
{R^{\rm{ZF}}} = \sum\limits_{k = 1}^{{N_u}} \mathbb{E}{\left\{ {{{\log }_2}\left( {1 + P{\rho _k^2}} \right)} \right\}}.
\end{equation}

From \eqref{ZF DL rate real}, the rate of the ZF precoder can be seen to be solely determined by the power normalization coefficients. In this paper, we consider two power normalization methods. The first is referred to as long-term normalization where the matrix $\bf{P}$ is adjusted by the long-term channel state information (CSI) and holds for the coherent time of the channel. In particular, the normalization coefficients are expressed as
\begin{equation}\label{ZF long-term nomalization}
\rho  = {\rho _1} = {\rho _2} =  \cdots  = {\rho _{{N_u}}} = \frac{1} {\sqrt{\mathbb{E}\left\{ {\left\| {\bf{\bar W}} \right\|_F^2} \right\}} }.
\end{equation}

Based on the achievable rate expression of the ZF precoder and the definition of long-term normalization, we derive the following theorem to provide its approximation.

\begin{theorem}\label{DL ZF1 appr}
When adopting the ZF precoder and the long-term normalization in the downlink of the DFT-based hybrid beamforming system, the achievable rate is approximated to
\begin{equation}\label{ZF DL rate appr lt}
R_{\rm{App}}^{\rm{ZF1}} = {N_u} {\log _2}\left( {1 + \frac{{P\left( {{N_s} - {N_u}} \right)}}{{\sum\limits_{k = 1}^{{N_u}} {\beta _k^{ - 1}{{\left[ {{{{\bf{\hat \Sigma }}}^{ - 1}}} \right]}_{k,k}}} }}} \right),
\end{equation}
where ${\bf{\hat \Sigma }}$ is defined in \eqref{central Wishart covariance}.
\end{theorem}

\begin{proof}
See Appendix \ref{Apdx:TheoDLZF1}.
\end{proof}

From \eqref{ZF DL rate appr lt}, it is found that for the ZF precoder, when the long-term normalization is adopted, all the users have equal received SINR and the rate of each user is same as well. It is not optimal because the user channels are different in quality.
Therefore, we further introduce the second method.

The second power normalization method is referred to as the short-term normalization where the power normalization coefficients are derived according to  the instantaneous channel information. In this method, we have
\begin{equation}\label{ZF short-term nomalization}
{\rho _k} = \frac{1}{{\sqrt {{N_u}} \left\| {\bf{\bar w}}_k \right\|}},
\end{equation}
which means that each data stream is allocated with equal transmit power. Moreover, it requires real-time CSI calculation and normalization factor adjustment.

\begin{theorem}\label{DL ZF2 appr}
When adopting the ZF precoder and the short-term normalization in the downlink of the DFT-based hybrid beamforming system, the achievable rate is approximated to
\begin{equation}\label{ZF DL rate appr st}
R_{\rm{App}}^{\rm{ZF2}} = \sum\limits_{k = 1}^{{N_u}} {{{\log }_2}\left( {1 + \frac{{P\left( {{N_s} - {N_u} + 1} \right)}}{{{N_u} {\beta _k^{-1}}{{\left[ {{{{\bf{\hat \Sigma }}}^{ - 1}}} \right]}_{k,k}}}}} \right)} .
\end{equation}
\end{theorem}

\begin{proof}
See Appendix \ref{Apdx:TheoDLZF2}.
\end{proof}

Let us now compare \eqref{ZF DL rate appr st} with \eqref{ZF DL rate appr lt}. Since the function $\log_2 \left( 1+ a x^{-1}\right)$, for $a>0$ is concave for $x>0$, when regarding ${\beta _k^{ - 1}} {{\left[ {{{{\bf{\hat \Sigma }}}^{ - 1}}} \right]}_{k,k}}$ as $x_k$, we can derive that $R_{\rm App}^{\rm ZF1} \le R_{\rm App}^{\rm ZF2}$ according to the Jensen's inequality. It validates our previous thoughts that short-term normalization performs better than long-term normalization for the ZF precoder.

Furthermore, following the example of the uplink, we give the asymptotic expressions of the achievable rates for the ZF precoder in the pure Ricean fading environments.

\begin{corollary}\label{ZF DL appr cor2}
When $K_k \to \infty$ for $k = 1,\dots,N_u$, if orthogonality holds among the effective LoS components from different users, \eqref{ZF DL rate appr lt} and \eqref{ZF DL rate appr st} approach, respectively, to
\begin{equation}\label{ZF DL cor2 approach}
R_{{\rm{App}}}^{\rm{ZF1}} \to {N_u} R_{{\rm{App}},k}^{\rm{ZF1}}, \quad R_{{\rm{App}}}^{\rm{ZF2}} \to \sum\limits_{k = 1}^{{N_u}}{R_{{\rm{App}},k}^{\rm{ZF2}}},
\end{equation}
where
\begin{equation}\label{ZF DL cor2 lt}
R_{{\rm{App}},k}^{\rm{ZF1}}= {\log _2}\left( {1 + \frac{{P\left( {{N_s} - {N_u}} \right)}}{{N_s}{\sum\limits_{i = 1}^{{N_u}} {\beta _i^{ - 1} {\left\| {{\bf{F}}{\bf{\bar h}}_i} \right\|}^{-2}} }}} \right)
\end{equation}
and
\begin{equation}\label{ZF DL cor2 st}
R_{{\rm{App}},k}^{\rm{ZF2}}= {{{\log }_2}\left( {1 + \frac{{P\left( {{N_s} - {N_u} + 1} \right)}}{{{N_s}{N_u}{\beta _k^{-1}}{\left\| {{\bf{F}}{\bf{\bar h}}_k} \right\|}^{-2}}}} \right)}.
\end{equation}
\end{corollary}

\begin{proof}
The proof is based on the fact that ${\left[ {{{{\bf{\hat \Sigma }}}^{ - 1}}} \right]}_{k,k} = \alpha _k^{-1} = \frac{N_s}{\left\| {{\bf{F\bar h}}_k} \right\|^2}$.
\end{proof}

\subsection{MRT Precoder}
With the MRT precoder in the downlink, ${\bf{W}}_D$ is found as
\begin{equation}\label{MRT precoder}
{{\bf{W}}_D} = {\bf{G}}_{eq}^{*}{\bf{P}},
\end{equation}
where ${\bf{P}}$ is previously defined as the power normalization matrix. The achievable rate is calculated as
\begin{equation}\label{MRT DL rate real}
{R^{\rm MRT}} = \sum\limits_{k = 1}^{{N_u}} \mathbb{E}{\left\{ {{{\log }_2}\left( {1 + \frac{{P \rho _k^2{{\left\| {{{\bf{g}}_{eq,k}}} \right\|}^4}}}{{P \sum\limits_{j \ne k} {\rho _j^2{{\left| {{{\bf{g}}_{eq,k}^H}{\bf{g}}_{eq,j}} \right|}^2}}  + {\rm{1}}}}} \right)} \right\}}.
\end{equation}
We first derive the achievable rate approximation for the long-term normalization.

\begin{theorem}\label{DL MRT1 appr}
Using the MRT precoder and the long-term normalization in the downlink, the achievable rate is approximated to
\begin{equation}\label{MRT DL rate appr lt}
\begin{aligned}
&R_{\rm{App}}^{\rm{MRT1}} = \\
&\sum\limits_{k = 1}^{{N_u}} {{\log }_2}\left( {1 + \frac{{\frac{{P\beta _k^2}}{{{{\left( {{K_k} + 1} \right)}^2}}}{\chi _1^{(k)}}}} {{\sum\limits_{j \ne k} \frac{{P{\beta _k}{\beta _j}}}{{\left( {{K_k} + 1} \right)\left( {{K_j} + 1} \right)}}{{\chi _{2,j}^{(k)}}} + \sum\limits_{i = 1}^{{N_u}}{{\frac{{{\beta _i}}}{{{K_i} + 1}}{\chi _3^{(i)}}}}}}} \right),
\end{aligned}
\end{equation}
where ${\chi _1^{(k)}}$, ${\chi _2^{(k)}}$ and ${\chi _3^{(k)}}$ are defined in \eqref{MRC UL rate appr x1}, \eqref{MRC UL rate appr x2} and \eqref{MRC UL rate appr x3}, respectively.
\end{theorem}

\begin{proof}
See Appendix \ref{Apdx:TheoDLMRT1}.
\end{proof}

Similarly, for the short-term normalization, the following theorem provides the approximation of the achievable rate.

\begin{theorem}\label{DL MRT2 appr}
When employing the MRT precoder and the short-term normalization in the downlink, the achievable rate is approximated as
\begin{equation}\label{MRT DL rate appr st}
R_{\rm{App}}^{\rm{MRT2}} = \sum\limits_{k = 1}^{{N_u}} {{\log }_2}\left( {1 + \frac{{\frac{{P{\beta _k}}}{N_u\left({{K_k} + 1}\right)}\chi _3^{(k)}}}{{\sum\limits_{j \ne k} {\frac{{P{\beta _k}}}{N_u\left({{K_k} + 1}\right)} \frac{\chi _{2,j}^{(k)}}{{\chi _3^{(j)}}}}}  + 1}} \right).
\end{equation}
\end{theorem}

\begin{proof}
See Appendix \ref{Apdx:TheoDLMRT2}.
\end{proof}

To compare the performance of the long-term and the short-term normalization for the MRT precoder, we derive the rate limits in pure Ricean fading conditions below.

\begin{corollary}\label{MRT DL appr cor2}
When $K_k \to \infty$ for $k = 1,\dots,N_u$, if orthogonality holds among the effective LoS components of different users, \eqref{MRT DL rate appr lt} and \eqref{MRT DL rate appr st} are approaching to, respectively,
\begin{equation}\label{MRT DL cor2 approach}
R_{{\rm{App}}}^{\rm{MRT1}} \to \sum\limits_{k = 1}^{{N_u}}{R_{{\rm{App}},k}^{\rm{MRT1}}},\quad R_{{\rm{App}}}^{\rm{MRT2}} \to \sum\limits_{k = 1}^{{N_u}}{R_{{\rm{App}},k}^{\rm{MRT2}}},
\end{equation}
where
\begin{equation}\label{MRT UL cor2 lt}
R_{{\rm{App}},k}^{\rm{MRT1}} = {{{\log }_2}\left( 1+ \frac{P \beta_k^2 {\left\| {{\bf{F\bar h}}_k} \right\|^4}}{\sum\limits_{i = 1}^{{N_u}} \beta_i {\left\| {{\bf{F\bar h}}_i} \right\|^2}}  \right)}
\end{equation}
and
\begin{equation}\label{MRT UL cor2 st}
R_{{\rm{App}},k}^{\rm{MRT2}} = {{{\log }_2}\left( 1+ \frac{P \beta_k {\left\| {{\bf{F\bar h}}_k} \right\|^2}}{N_u}  \right)}.
\end{equation}
\end{corollary}

For either $K_k = 0$ or $K_k \to \infty$ cases, if the short-term normalization is employed with the MRT precoder, users with lower $\beta_k$ can obtain relatively higher received SINR when compared with the long-term normalization. On the contrary, for the long-term normalization, users with higher $\beta_k$ will be allocated with more power on their data streams. Thus, we can conclude that for MRT precoders, the long-term normalization is preferred for the users who have stronger channel quality, and that the short-term normalization improves the quality of users with poor propagation conditions. The situation will be different in the case of the ZF precoder.

\section{Analog Beam Selection}\label{Sec:Beam Selection}

The asymptotic expressions and the analytical results above help us develop new beam selection solutions for the DFT-based hybrid multiuser system. Here, we study how to design the analog beamformer by utilizing these analytical results and present three beam selection schemes to pursue high achievable rate, which are realized by the designs of ${\bf{\Psi }}$.

\subsection{Achievable Rate Based Exhaustive Searching}
One can aim to achieve the highest rate for the multiuser system based on the previously derived approximations. According to the directions of transmission link, the type of precoders or receivers and the normalization method, we can use the appropriate approximation to help select the DFT beams. This can be done by searching over all the possible beam combinations to identify the joint optimum beams for all the RF chains. We refer to this scheme as the achievable rate based exhaustive searching scheme.

\begin{table}
  \label{tab:exhausted searching algorithm}
  \centering
  \begin{tabular}{l}
    \hline
    \bfseries Algorithm 1 Exhaustive Searching \\
    \hline
    \bfseries Require: ${\bf{\Psi }}$ \\
    1: set $R_{\rm max} = 0$ \\
    2: for $i_1,\dots,i_{N_s} \le M$ do \\
    3: ~ ${\bf{\Psi }}_{tmp} = \left[{\bf e}_{i_1},\dots,{\bf e}_{i_{N_s}}\right]$ \\
    4: ~ calculate $R \left( {\bf{\Psi }}_{tmp}\right)$ using the corresponding\\
    \qquad ~ approximation expression\\
    5: ~ if $R \ge R_{\rm max}$ \\
    6: ~~~ ${\bf{\Psi }} = {\bf{\Psi }}_{tmp}$ \\
    7: ~ end if \\
    8: end for \\
    \bfseries return ${\bf{\Psi }}$ \\
    \hline
  \end{tabular}
\end{table}

As mentioned above, each of the ${N_s}$ RF chains will be assigned with a DFT beam selected from the $M$-sized codebook. There are totally $M^{{N_s}}$ beam combinations. Assume that $\bf D$, $\bf \Omega$ and $\bar{\bf H}$ are known at the BS. Let us take the uplink ZF receiver as an example. After calculating the achievable rate of each beam combination, we get the optimum combination that maximizes \eqref{ZF UL rate appr1} by solving
\begin{subequations}
\begin{equation}\label{exhausted searching goal}
\mathop {\max} \limits_{{\bf{F}}} {R_{{\rm{App}}}^{\rm{ZF}}}
\end{equation}
\begin{equation}\label{exhausted searching st1}
{\rm s.t.}\quad{\bf F} = {\bf\Psi}{\bf U}, {\bf{\Psi}} = \left[ {\bf{e}}_{i_1}, {\bf{e}}_{i_2},\dots, {\bf{e}}_{i_{N_s}} \right]^T,
\end{equation}
\begin{equation}\label{exhausted searching st2}
i_1,i_2,\dots,i_{N_s} = 1,2,\dots,M.
\end{equation}
\end{subequations}
The approach is formulated as Algorithm 1.

The achievable rate based exhaustive searching scheme uses the derived theorems and strives for the optimization of global achievable rate performance. Therefore, it obtains the optimum beam selection results and achieves the highest ergodic rate. However, it is incredibly time-consuming when either $M$ or ${N_s}$ grows large, which can become infeasible quickly.

\subsection{Projected Power Based Per-user Selection}
Considering the drawback of exhaustive searching, there is need to explore other suboptimal schemes which are practically more feasible. The first thing is to abandon exhaustive searching which requires many power-level comparisons.
To do so, we refer to \emph{Corollaries \ref{ZF UL appr1 cor2}, \ref{MRC UL appr cor3}, \ref{ZF DL appr cor2}} and \emph{\ref{MRT DL appr cor2}} for simpler expressions of the achievable rate.

\begin{table}
  \label{tab:per-user selection algorithm}
  \centering
  \begin{tabular}{l}
    \hline
    \bfseries Algorithm 2 Per-user Selection \\
    \hline
    \bfseries Require: ${\bf{\Psi }}$ \\
    1: ${\bf{\Psi }}=$ empty matrix \\
    2: for $k \le {N_u}$ do \\
    3: ~ $\left[ {j_1^{(k)},\dots,j_C^{(k)}} \right] = \mathop {\arg \max }\limits_{n = 1,\dots,M} \left( {{{\bf{U}}^H}{\bf{\bar h}}_k}{\bf{\bar h}}_k^H{\bf{U}} \right)_{k,k}$ \\
    4: ~ for $C\left( {k - 1} \right) + 1 \le m \le \min \left( {Ck,{N_s}} \right)$ do\\
    5: ~~~ ${i_m} = j_{m - C\left( {k - 1} \right)}^{(k)}$ \\
    6: ~~~ ${\bf{\Psi }}\left( {:,m} \right) = {{\bf{e}}_{{i_m}}}$ \\
    7: ~ end for \\
    8: end for \\
    \bfseries return ${\bf{\Psi }}$ \\
    \hline
  \end{tabular}
\end{table}

To do so, according to \emph{Corollaries \ref{ZF UL appr1 cor2}, \ref{MRC UL appr cor3}, \ref{ZF DL appr cor2}} and \emph{\ref{MRT DL appr cor2}}, the  rate limits have much simpler expressions for analysis. Take \emph{Corollary \ref{ZF UL appr1 cor2}} as an example.
We find that for user $k$,
\begin{equation}\label{per-user selection goal1}
\mathop {\max }\limits_{\bf F}{R_{{\rm{App}},k}^{\rm{ZF1}}}, \quad {\rm s.t.}\quad \eqref{exhausted searching st1}\quad {\rm and} \quad \eqref{exhausted searching st2},
\end{equation}
can be recast into
\begin{equation}\label{per-user selection goal2}
\mathop {\max }\limits_{\bf F}{{\left\| {{\bf{F}}{\bf{\bar h}}_k} \right\|}^2}, \quad{\rm s.t.} \quad \eqref{exhausted searching st1} \quad{\rm and}\quad \eqref{exhausted searching st2}
\end{equation}
when ${N_s},{N_u}$ and ${\bf{\bar h}}_k$ are fixed.
It is because that the achievable rate is improved with the increase of ${\left\| {{\bf{F}}{\bf{\bar h}}_k} \right\|}^2$ under strong Ricean fading conditions, as is mentioned in the insights from \emph{Corollary \ref{ZF UL appr1 cor2}}. ${\left\| {{\bf{F}}{\bf{\bar h}}_k} \right\|}^2$ represents the projected power of the LoS component of user $k$ on the selected DFT beams. The enhancement of ${\left\| {{\bf{F}}{\bf{\bar h}}_k} \right\|}^2$ reflects that the selected beams are more competent to capture the main lobes of the LoS paths. Moreover, \eqref{per-user selection goal2} can be achieved without power-level comparisons, which is more time-saving.

Based on the analysis, we first introduce a per-user selection scheme to maximize the projected power of the LoS paths. For fairness, we strive to make each user use an equal number of RF chains. Denote $C = \left\lfloor {{{{N_s}} \mathord{\left/
{\vphantom {{{N_s}} {{N_u}}}} \right. \kern-\nulldelimiterspace} {{N_u}}}} \right\rfloor $. Then each of the first ${N_u} - 1$ users is allocated with $C$ RF chains, and the ${N_u}$th user is allocated with ${N_s} - C\left( {{N_u} - 1} \right)$ RF chains. For user $k$, we use ${\bf{\bar h}}_k$ to select beams. The steps for this beam selection approach are presented in Algorithm 2.

This projected power based per-user selection scheme reduces the number of comparisons from $M^{{N_s}}$ to $M \times N_u$. However, the orthogonality among the effective LoS paths from different users is not promised in this scheme, which will significantly impact the achievable rate performance.

\subsection{The Proposed Two-Step Selection Scheme}
Both the exhaustive searching method and the projected power based per-user selection have advantages and disadvantages. The exhaustive searching scheme insures the optimality of the selected beam combinations, including the capture of the LoS paths and the orthogonality among the effective LoS paths from different users. For the per-user selection scheme, it uses a simple metric, i.e., maximizing the projective power of the LoS paths, to implement beam selection. We can jointly utilize these two schemes to devise new and better solutions.

For the per-user selection scheme, beams are directly determined after comparing the projected power. By doing so, however, it is unclear if there exists severe interference among the effective LoS paths from different users. As a result, there needs to be a fix in the per-user selection. Since the CSI is known at the BS, we can evaluate any beam combination for any number of RF chains using the analytically derived approximations. The fact that the increase of the number of RF chains can enhance the performance of the hybrid beamforming system, suggests that in the per-user selection scheme, we should consider more than $N_s$ RF chains to choose more than $N_s$ beams. Then the extra beams can later be dropped using the approximations. This is the rationale of the proposed two-step selection scheme as Algorithm 3.

As seen, the scheme chooses more beams than needed in Step 1. There are $C+n$ beams selected for each user, where $n$ is a positive integer that directly determines the number of extra beams. Totally $N_u \left( C+n\right)$ beams are selected, and a big ${\bf{\Psi }}$ is derived in Step 1. Then in Step 2, the approximations in {\em Theorems} are adopted to help remove $N_u \left( C+n\right) - N_s$ extra beams one by one. It requires $N_u \left( C+n\right) - N_s$ rounds and one beam is removed at the end of each round. At the beginning of the $i$th round, there remain $N_u \left( C+n\right) -i+1$ beams. We calculate the rate when each of these beams is removed, i.e., $R_i\left( j\right) $ for the $i$th beam removed. Then, find the one that has the least performance reduction and drop it from ${\bf{\Psi }}$. After $N_u \left( C+n\right) - N_s$ rounds, we obtain the desired ${\bf{\Psi }}$.

\begin{table}
  \label{tab:two-step selection algorithm}
  \centering
  \begin{tabular}{l}
    \hline
    \bfseries Algorithm 3 Two-Step Selection \\
    \hline
    \bfseries Require: ${\bf{\Psi }}$ \\
    1:  ${\bf{\Psi }}=$ empty matrix \\
    2:  Step 1:\\
    3:  for $k \le {N_u}$ do \\
    4:  ~ $\left[ {j_1^{(k)},\dots,j_C^{(k)}} \right] = \mathop {\arg \max }\limits_{n = 1,...,M} \left( {{{\bf{U}}^H}{\bf{\bar h}}_k}{\bf{\bar h}}_k^H{\bf{U}} \right)_{k,k}$ \\
    5:  ~ for $\left( C+n\right)\left( {k - 1} \right) + 1 \le m \le \left( C+n\right)k$ do\\
    6:  ~~~ ${i_m} = j_{m - C\left( {k - 1} \right)}^{(k)}$ \\
    7:  ~~~ ${\bf{\Psi }}\left( {:,m} \right) = {{\bf{e}}_{{i_m}}}$ \\
    8:  ~ end for \\
    9:  end for \\
    10: Step 2: \\
    11: for $1 \le i \le {N_s}$ do \\
    12: ~ for $1 \le j \le N_u \left( C+n\right) -i+1$ do\\
    13: ~~~ Calculate $R_i\left( j\right) $ with $i$th column of ${\bf{\Psi }}$ removed\\
    14: ~ end for\\
    15: ~ $i_{\rm{remove}} = \mathop {\arg \max } R_i\left( j\right)$\\
    16: ~ ${\bf{\Psi }}$ with the $i_{\rm{remove}}$th column removed\\
    17: end for\\
    \bfseries return ${\bf{\Psi }}$ \\
    \hline
  \end{tabular}
\end{table}

It should be noted that we do not suggest to choose less than $N_s$ beams and then add new beams in. This is because under Ricean fading, the channel holds sparsity if the Ricean $K$-factor grows large. A small number of beams can capture the main lobe of the channel. If we assume there are more RF chains and choose more beams at first, then the main lobe will be totally captured and the inter-user interference can be thoroughly canceled. These beams are adequate and we only need to remove the ones which make insignificant contributions. However, if we assume there are less RF chains and choose less beams, then the LoS paths cannot be completely covered and the interference cannot be completely eliminated, which will jeopardize the results from the first step.

Note that for the analog beamformer ${\bf F} = \left[ {\bf{f}}_1,\dots,{\bf{f}}_{N_s} \right]$, the difference made by changing the order of ${\bf{f}}_1,\dots,{\bf{f}}_{N_s}$ is negligible, but it requires much more comparisons to find the optimal order according to the exhaustive searching rule. Hence, to avoid expensive and time-consuming comparisons, it is not a good idea to follow the exhaustive searching rule in Step 2, but to keep the order of the beams derived in Step 1. For example, $N_s=4, N_u=2, n=1$, and the selected beams in Step 1 are ${\bf f}_1,{\bf f}_2,{\bf f}_3$ for user 1 and ${\bf f}_4,{\bf f}_5,{\bf f}_6$ for user 2. Then the order of the beams derived in Step 1 is $1,2,3,4,5,6$. In Step 2, this order will not be changed. If we evaluate the performance when ${\bf f}_2$ is removed, then the only analog beamformer to be evaluated will be ${\bf F} = \left[ {\bf{f}}_1,{\bf{f}}_3,{\bf{f}}_4,{\bf{f}}_5,{\bf{f}}_6 \right]$. Other beam orders such as $1,1,3,3,5$ or $3,4,1,6,5$ will not be in consideration. Therefore, a large amount of power-level comparisons caused by the exhaustive searching are avoided.

\begin{table*}[!t]
  \caption{The Numbers of Comparisons of the Schemes}
  \label{tab:selection time}
  \centering
  \begin{tabular}{|l|l|}
    \hline
    Exhaustive Searching & $M^{{N_s}}$ \\
    \hline
    Per-user Selection & $M \times N_u$ \\
    \hline
    Two-step Selection & $M \times N_u +\frac{1}{2} \left[{N_u^2 \left( C+n\right)^2}+N_u \left( C+n\right)-N_s^2-N_s\right]$\\
    \hline
  \end{tabular}
\end{table*}

Comparisons required by the two-step selection have two parts. The first part contains the same number of comparisons as the per-user selection scheme, $M \times N_u$, which are caused by the power projection in Step 1. The second part includes the rounds to remove the extra beams, which requires
\begin{equation}\label{two-step selection time}
\begin{aligned}
&\sum\limits_{i = 1}^{{N_u}\left(C+n\right)-N_s} {\left[N_u \left( C+n\right) -i+1\right]} \\
=& \frac{1}{2} \left[{N_u^2 \left( C+n\right)^2}+N_u \left( C+n\right)-N_s^2-N_s\right]
\end{aligned}
\end{equation}
comparisons. We realize that increasing $n$ improves the final selection results, but increases the number of comparisons at the same time. The total number of comparisons required by the two-step selection is given in Table \ref{tab:selection time}. With $M=256, N_s=8, N_u=4$ and $n=2$, the numbers of comparisons for exhaustive searching, per-user selection and two-step selection are $1.8\times 10^9$, $1024$ and $1124$, respectively.

It is worth mentioning that when compared with non-codebook-based hybrid beamforming designs, the analog beam selection schemes have relatively inferior performance due to the codebook constraint in RF module. However, the complexity of the non-codebook-based design is generally much higher than the codebook-based design. For example, non-codebook-based methods are usually realized through phase shifter networks. We need to adjust all the phase shifters if the beamforming weights are updated, which further causes more power consumption and brings challenge to the synchronization of the phase shifter network. While for the Butler matrix, the beamforming weights are formed once the Butler matrix is implemented on the hardware, and afterwards the only thing is to switch among these beams. Considering both the advantages and disadvantages, it will be a low-cost but efficient way to realize hybrid beamforming by adopting Butler matrix structure and the analog beam selection schemes.

\section{Numerical Results}\label{Sec:Numerical Results}
To validate the derived uplink and downlink achievable rate approximations in \emph{Theorems \ref{UL ZF appr1}--\ref{DL MRT2 appr}} and evaluate the performance of the proposed analog beam selection schemes for this DFT-based hybrid beamforming multiuser system, we here conduct computer simulations and discuss the numerical results.


\begin{figure}
  \centering
  \includegraphics[scale=0.48]{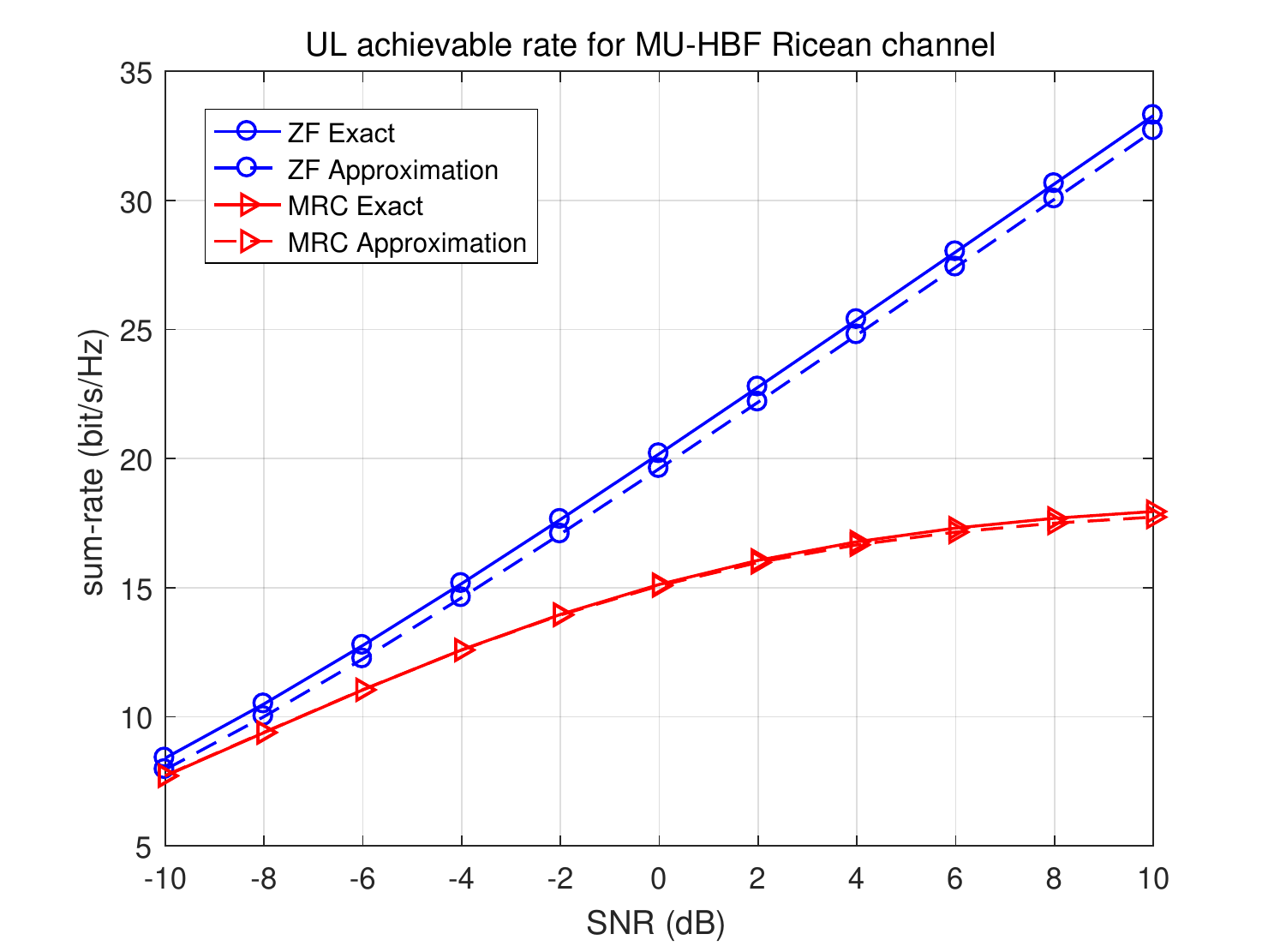}
  \caption{Uplink achievable rates versus SNR for the DFT-based hybrid beamforming multiuser system, with $M=512, N_s=32, N_u=4, K_1=K_2=...=K_{N_u}=10$dB and the two-step selection adopted.}
  \label{UL SNR evaluation}
\end{figure}
\begin{figure}
  \centering
  \includegraphics[scale=0.48]{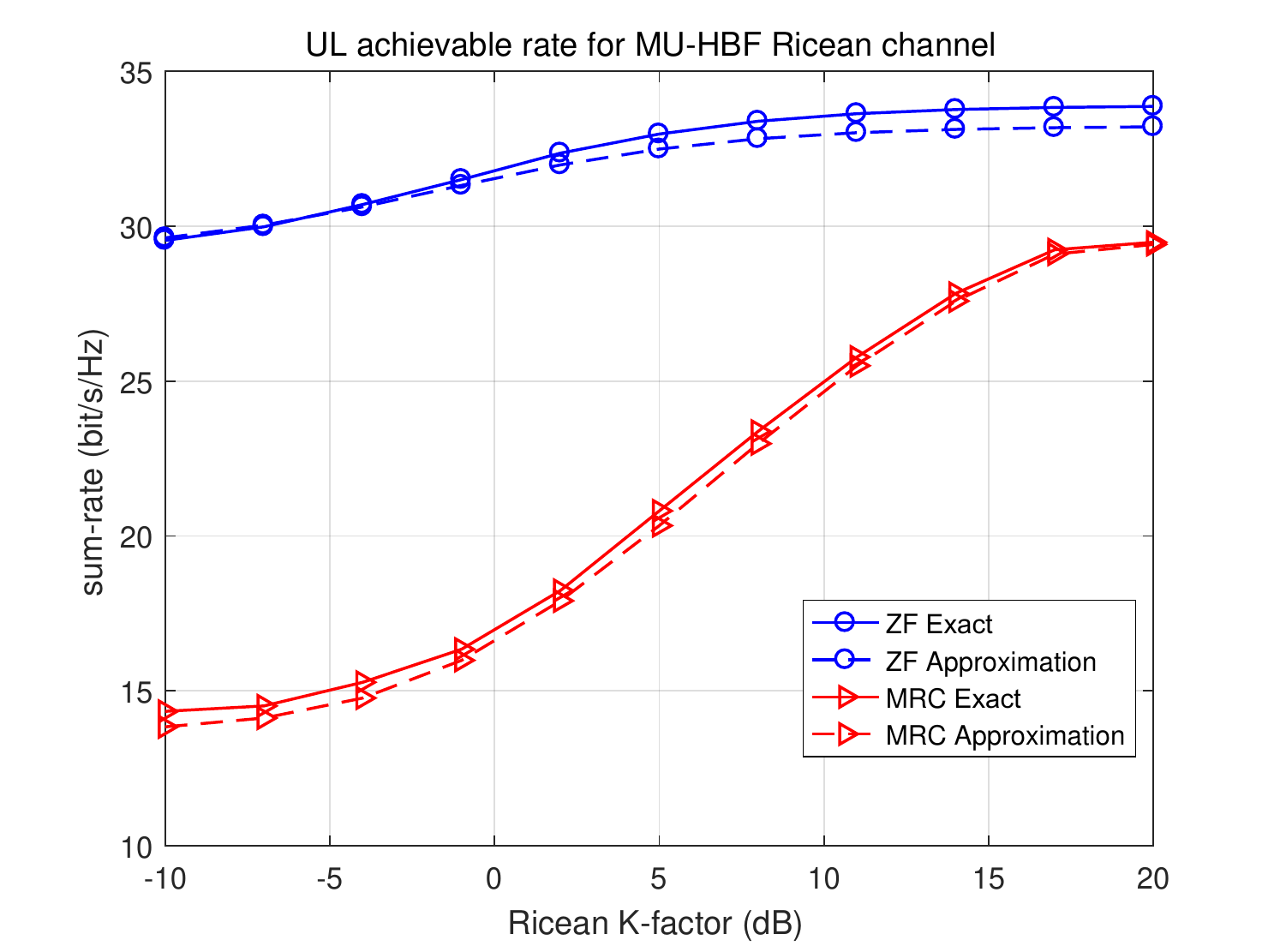}
  \caption{Uplink achievable rates versus Ricean $K$-factor for the DFT-based hybrid beamforming multiuser system, with $M=512, N_s=32, N_u=4, SNR=10$dB and the two-step selection adopted.}
  \label{UL K evaluation}
\end{figure}

Fig.~\ref{UL SNR evaluation} compares the Monte Carlo (exact) and approximation results of the uplink achievable rates. In order to save the processing time, the two-step selection is adopted. Here, signal-to-noise ratio (SNR) measures the uplink transmit power of each user against the noise power on each BS antenna. In the simulations, we set $M=512, N_s=32, N_u=4$. For convenience, the Ricean $K$-factors of the user are set to be equal as $K_1=K_2=\cdots=K_{N_u}=10$dB. The elements of $\bar{\bf H}$ are i.i.d.~and generated with zero mean and unit variance and fixed during a statistical period of 1000 drops. From the results in Fig.~\ref{UL SNR evaluation}, we can see that the ZF approximation approaches the exact results closely, and the MRC approximation almost coincides with the exact results. These results strongly validate the effectiveness of both the ZF and the MRC approximations. Also, we observe that in the low SNR regime, it is the noise that impacts the achievable rate most. The MRC receiver combines and improves the received power of the target signal; therefore it has competitive behavior with the ZF receiver. As SNR increases, the power of both the target signal and the inter-user interference increases. ZF effectively eliminates the interference and performs far more better than the MRC receiver. Since both the ZF receiver and the MRC receiver are susceptible to noise and interference, respectively, their gap in performance becomes wider with the increase of SNR.

Fig.~\ref{UL K evaluation} examines the influence of Ricean $K$-factor on the achievable rate. The simulations were conducted with the same configuration as in Fig.~\ref{UL SNR evaluation} except that SNR was set to 10dB. From Fig.~\ref{UL K evaluation}, it can be clearly observed that when the Ricean $K$-factor is small, the channel becomes more like Rayleigh distributed, and the rate performance is poor. With the increase of the Ricean $K$-factor, both the ZF receiver and the MRC receiver achieve higher achievable rate, especially for the MRC receiver. This is because on the one hand, the increase of the Ricean $K$-factor reflects greater dominance of the LoS path as well as the lower channel sidelobes, which further contributes to less inter-user interference; on the other hand, according to the analysis in Section \ref{Sec:Uplink Rate}, the achievable rate in pure Ricean fading conditions is not always higher than that under Rayleigh fading conditions except that the selected analog beams capture the LoS paths and contribute to the orthogonality among the analog beamformed LoS paths from different users. The improvement of the achievable rate reflects the effectiveness of the two-step beam selection results.

\begin{figure}
  \centering
  \includegraphics[scale=0.48]{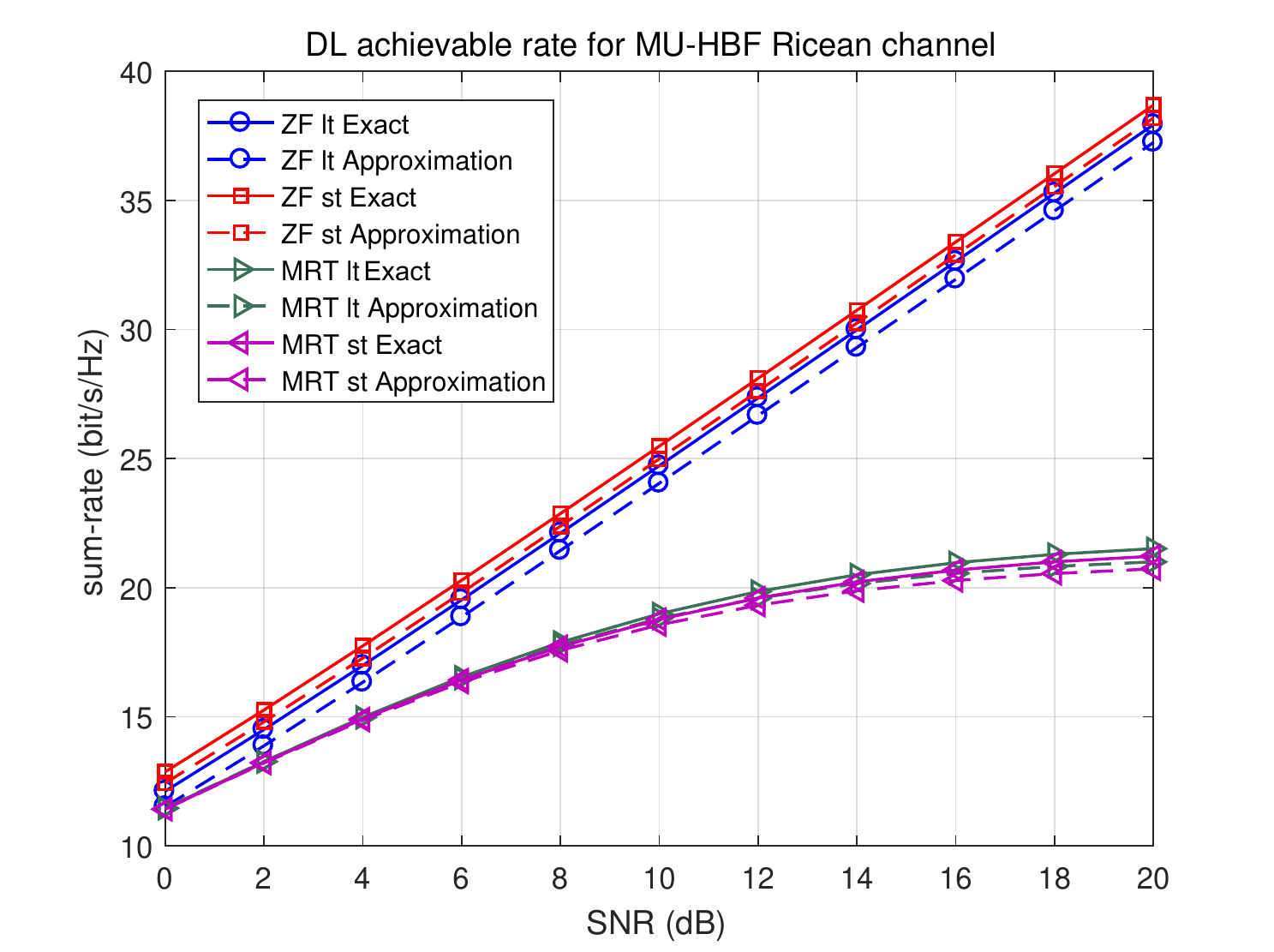}
  \caption{Downlink achievable rate versus SNR for the DFT-based hybrid beamforming multiuser system, with $M=512, N_s=32, N_u=4, K_1=K_2=...=K_{N_u}=10$dB and the two-step selection adopted.}
  \label{DL SNR evaluation}
\end{figure}
\begin{figure}
  \centering
  \includegraphics[scale=0.48]{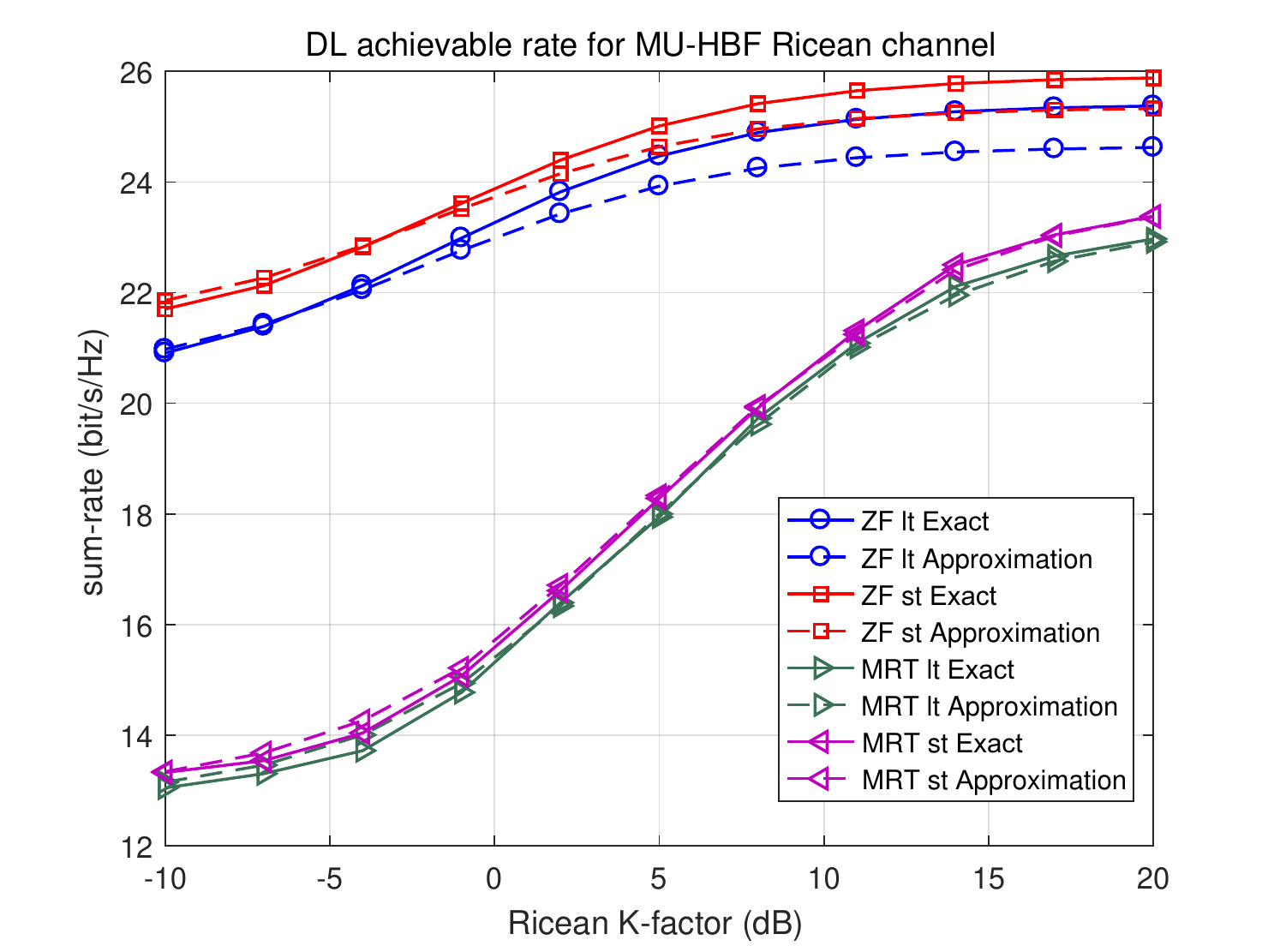}
  \caption{Downlink achievable rate versus Ricean $K$-factor for the DFT-based hybrid beamforming multiuser system, with $M=512, N_s=32, N_u=4, SNR=10$dB and the two-step selection adopted.}
  \label{DL K evaluation}
\end{figure}

For the downlink, Fig. \ref{DL SNR evaluation} and Fig. \ref{DL K evaluation} give the exact and the approximation results of the achievable rate versus SNR and Ricean $K$-factor, respectively. Simulation conditions are the same as that used in Fig.~\ref{UL SNR evaluation} and Fig. \ref{UL K evaluation} accordingly. It should be noted that the downlink SNR represents the total transmit power at the BS side against the noise power on each user antenna. A close observation from the figures reveals that the approximations are close to the exact results, demonstrating that the downlink approximations are valid. Similarly, the ZF precoder outperforms the MRT precoder in the high SNR regime, and the achievable rate of both these two precoders is proportional to the Ricean $K$-factor. When it comes to the normalization methods, we can see that for the ZF precoder the short-term normalization always achieves higher rate than the long-term normalization, while for the MRT precoder the two normalization methods have mixed performance since each of them has its own advantages, which justifies our previous analysis. Furthermore, we also see that for the ZF receiver/precoder, the approximations are less tight due to the loose central Wishart approximation when Ricean $K$-factors are large but the difference between $N_s$ and $N_u$ is small.

\begin{figure}
  \centering
  \includegraphics[scale=0.48]{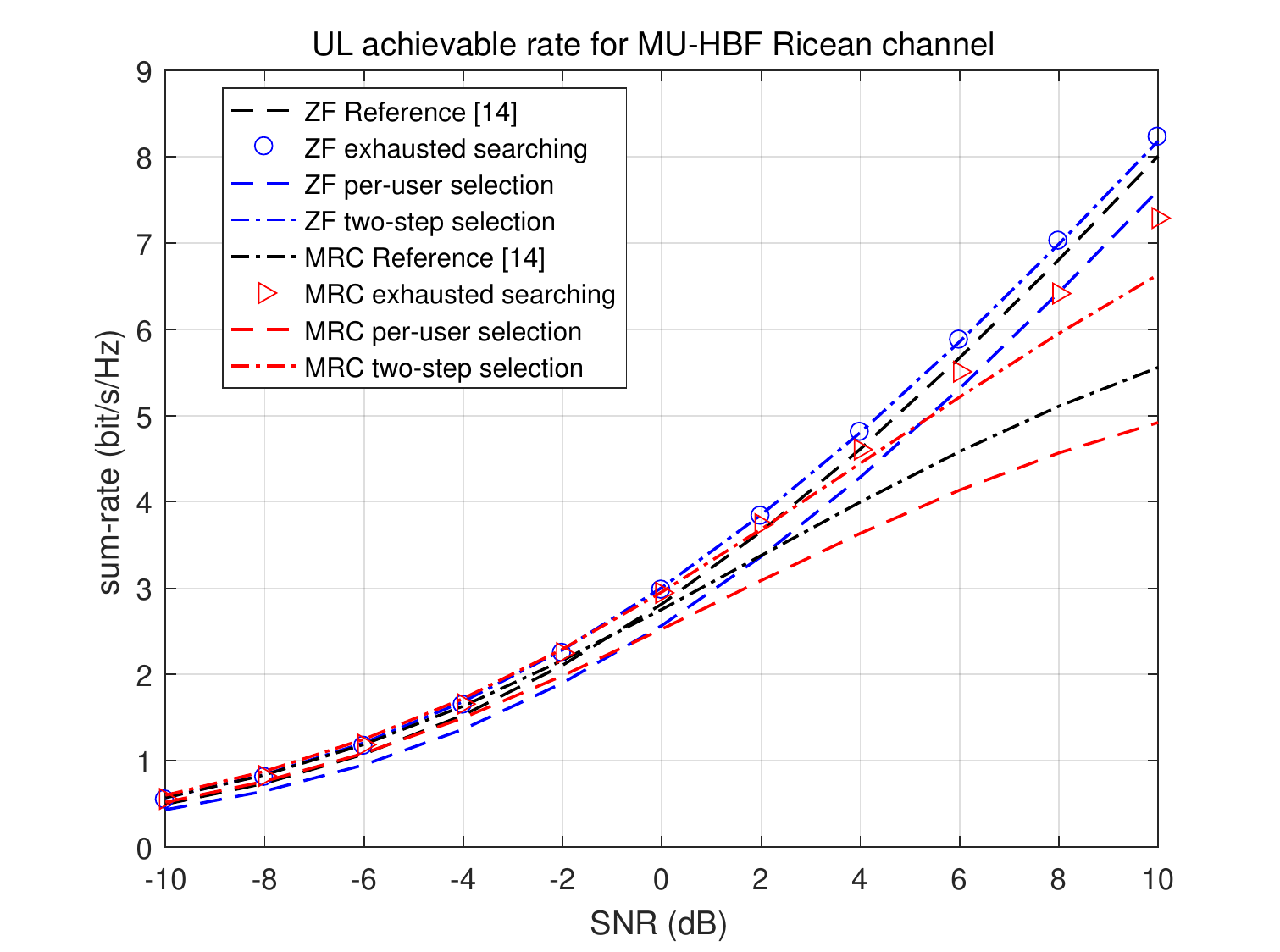}
  \caption{Comparison of the analog beam selection schemes in the uplink, with $M = 128,{N_s} = 4, n = 1, {N_u} = 2$, and $ K_1=K_2=...=K_{N_u}=10$dB.}\label{UL F evaluation}
\end{figure}
\begin{figure}
  \centering
  \includegraphics[scale=0.48]{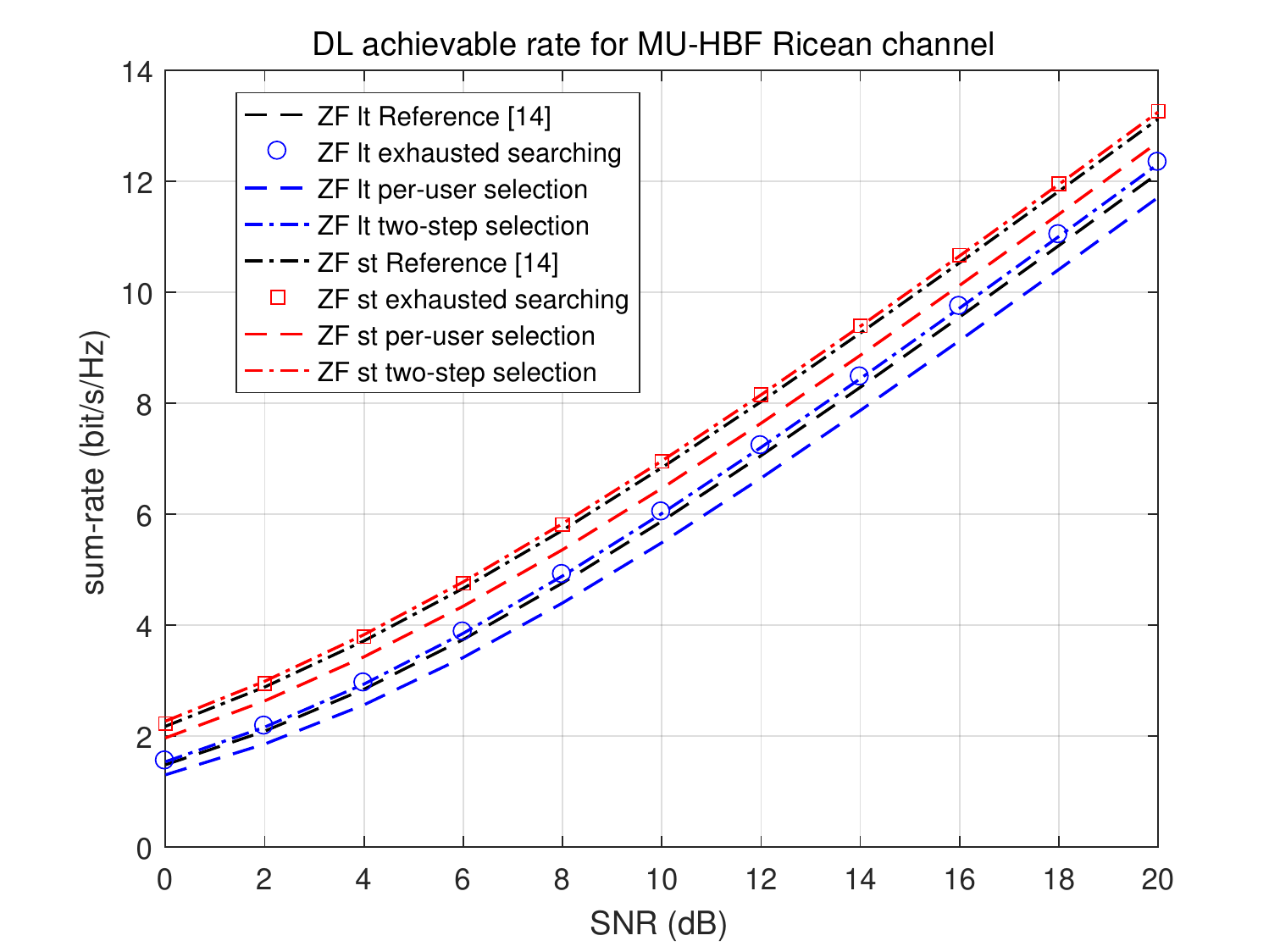}
  \caption{Comparison of the analog beam selection schemes in the downlink when ZF precoder is adopted, with $M = 128,{N_s} = 4, n = 1, {N_u} = 2$, and $ K_1=K_2=...=K_{N_u}=10$dB.}\label{DL F evaluation1}
\end{figure}

Next, we examine the exhaustive search, the per-user selection and the two-step selection through Monte Carlo simulations.
Fig.~\ref{UL F evaluation} compares the uplink achievable rate of the proposed three selections with the two-stage multiuser hybrid precoder introduced in \cite{Alkhateeb2015}. The two-stage multiuser hybrid precoder first chooses analog beams from a codebook through the downlink training process and then calculates the digital beamformer with low-dimensional CSI that is fed back to the BS. When implementing two-stage hybrid precoders in the simulations, we adopt DFT codebooks and ZF/MRC/MTR receivers/precoders as well, and assume that CSI is perfectly sent back to the BS. Considering the implementation of the exhaustive search, we set $M = 128,{N_s} = 4,{N_u} = 2$ and $K_1=K_2=...=K_{N_u}=10$dB. The margin for the two-step selection scheme is set to $n=1$. An obvious performance gap can be seen between the exhaustive search and the per-user selection, especially for the MRC receiver. The reason is that on the one hand, the per-user selection only guarantees the capture of the LoS paths, without considering the interference among the analog beamformed LoS paths from different users. On the other hand, even if the interference is not completely eliminated by the exhaustive search based analog beamforming, the MRC receiver behaves much more sensitive to the interference than the ZF receiver. If the two-step selection is adopted, the achievable rate will become much closer to that of the exhaustive search. The rate improvement is significant for the MRC receiver, which demonstrates that the beams chosen by the two-step selection scheme are more effective for interference cancellation. The two-stage multiuser hybrid precoder behaves better than the per-user selection because it utilizes short-term CSI instead of long-term CSI. However, its performance is inferior to the two-step selection since the latter takes advantage of the rate approximations and benefits from a bigger candidate set selected in the first step. Moreover, the numbers of comparisons required by exhaustive search, per-user selection and two-step selection are ${128^4}$, 256 and 282. If we further increase $M$ or $N_s$, the exhaustive search will be incredibly time-consuming, while the number of comparisons required by the two-step selection increases only slightly. Fig.~\ref{DL F evaluation1} and Fig.~\ref{DL F evaluation2} illustrate the downlink comparison results of these three selections, which leads to similar insights as the uplink results. We can now conclude that the two-step selection is a near optimal scheme with low complexity.

\begin{figure}
  \centering
  \includegraphics[scale=0.48]{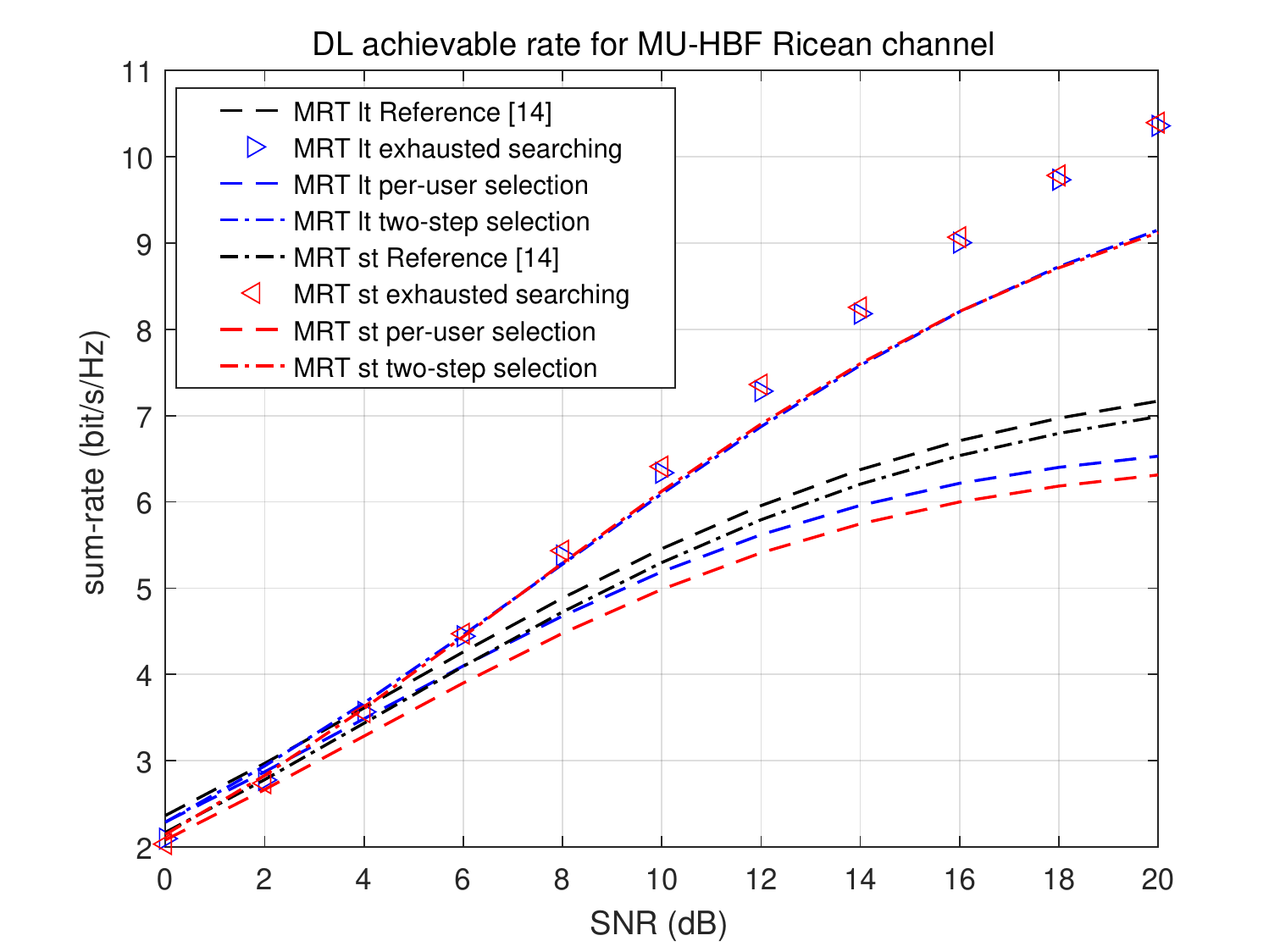}
  \caption{Comparison of the analog beam selection schemes in the downlink when MRT precoder is adopted, with $M = 128,{N_s} = 4, n = 1, {N_u} = 2$, and $ K_1=K_2=...=K_{N_u}=10$dB.}\label{DL F evaluation2}
\end{figure}

To better approach the performance of exhaustive search, we can keep more margin in the first step, that is, increase the value of $n$. Fig.~\ref{TwoStep n evaluation} illustrates the improvement of the achievable rate when $n$ is set from 1 to 2 and the long-term normalization is employed by the MRT precoder. This improvement comes from the better separation among the analog beamformed LoS paths from different users. If we increase the margin $n$, we can select beams from a more complete beam subset and enhance the effectiveness of the selection results. Furthermore, when $n$ increases from 1 to 2, the number of comparisons required by the two-step selection increases from 267 to 282, with only 15 comparisons more. Therefore, we can harvest significant performance enhancement with little more cost.

\begin{figure}
  \centering
  \includegraphics[scale=0.48]{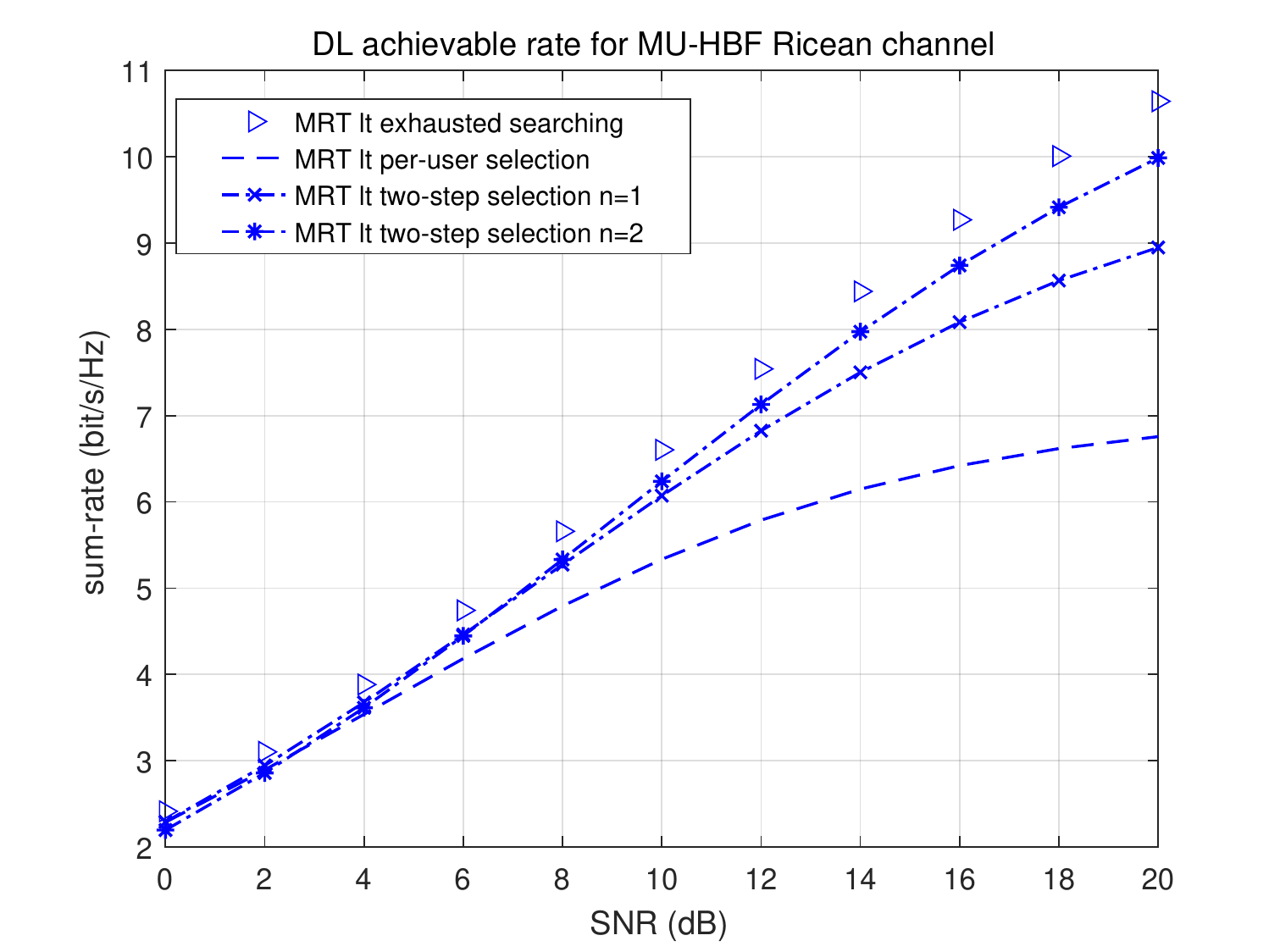}
  \caption{Evaluation of the margin $n$ for the downlink MRT precoder when the long-term normalization is adopted, with $M = 128,{N_s} = 4, {N_u} = 2$, and $ K_1=K_2=...=K_{N_u}=10$dB and $n=1,2$.}\label{TwoStep n evaluation}
\end{figure}

\section{Conclusion}\label{Sec:Conclusion}
This paper studied the analog beam selection schemes for the DFT-based hybrid beamforming multiuser system. For both uplink and downlink, we analyzed the achievable rates of the system using the ZF/MRC receivers and the ZF/MRT precoders considering long-term and short-term downlink normalization methods. Based on our approximations and asymptotic expressions of the achievable rates, we presented three analog beam selection schemes. The first one is the exhaustive searching scheme which is the optimal solution but with huge time-consumption. To avoid  power-level comparisons, we then proposed the projected power based per-user selection. Then we further proposed the two-step selection scheme which can obtain near-optimal results and is much more time-saving. Simulation results demonstrated that the asymptotic analysis is effective and the performance of the two-step selection approaches to that of exhaustive searching.

\appendix

\subsection{Proof of Theorem \ref{UL ZF appr1}}\label{Apdx:TheoULZF1}

\begin{figure*}[!b]
\normalsize
\newcounter{MYtempeqncnt}
\setcounter{MYtempeqncnt}{\value{equation}}
\setcounter{equation}{80}
\hrulefill
\begin{equation}\label{MRC UL rate appr x1 proof}
\mathbb{E}\left\{ {{{\left\| {{{\bf{g}}_{eq,k}}} \right\|}^4}} \right\}
\!=\! \frac{{\beta _k^2}}{{{{\left( {{K_k} \!+\! 1} \right)}^2}}} \left[ {K_k^2 {{ \left\| {{\bf{F\bar h}}_k} \right\|}^4} \!+\!{2{K_k}\left({{N_s} \!+\! 1} \right){{\left\| {{\bf{F\bar h}}_k} \right\|}^2}}}{ \!+\! {N_s} \left( {{N_s} \!+\! 1} \right)} \right]\\
\!=\! \frac{{\beta _k^2}}{{{{\left( {{K_k} \!+\! 1} \right)}^2}}} \left[ {\chi _3^{(k)2} \!+\! 2\chi _3^{(k)} \!-\! {N_s}} \right],
\end{equation}
\setcounter{equation}{\value{MYtempeqncnt}}
\setcounter{equation}{81}
\hrulefill
\begin{equation}\label{MRC UL rate appr x2 proof}
\begin{aligned}
\mathbb{E}\left\{ {{{\left| {{\bf{g}}_{eq,k}^H{{\bf{g}}_{eq,j}}} \right|}^2}} \right\}
& = \frac{{{\beta _k}{\beta _j}}}{{\left( {{K_k} + 1} \right)\left( {{K_j} + 1} \right)}}\left( {{K_k}{K_j}{{\left| {{\bf{\bar h}}_j^H{{\bf{F}}^H} {\bf{F}}{{{\bf{\bar h}}}_k}} \right|}^2} + {K_k}{{\left\| {{\bf{F\bar h}}_k^{}} \right\|}^2} + {K_j}{{\left\| {{\bf{F\bar h}}_j^{}} \right\|}^2} + {N_s}} \right)\\
& = \frac{{{\beta _k}{\beta _j}}}{{\left( {{K_k} + 1} \right)\left( {{K_j} + 1} \right)}}\left( {{K_k}{K_j}{{\left| {{\bf{\bar h}}_j^H{{\bf{F}}^H} {\bf{F}}{{{\bf{\bar h}}}_k}} \right|}^2} + {K_j}{{\left\| {{\bf{F\bar h}}_j^{}} \right\|}^2} + \chi _3^{(k)}} \right).
\end{aligned}
\end{equation}
\setcounter{equation}{\value{MYtempeqncnt}}
\end{figure*}

Since ${{\bf{G}}_{eq}} = {\bf{FG}}$ and $\bf{G} = {\bf{HD}}^{\frac{1}{2}}$, we can write a detailed expression of the uplink achievable rate of the ZF receiver as
\begin{equation}\label{ZF UL rate real2}
{R^{\rm ZF}} = \sum\limits_{k = 1}^{{N_u}} \mathbb{E}{\left\{ {{{\log }_2}\left( {1 + \frac{{{P_{avg}}{\beta _k}}}{{{{\left[ {{{\left( {{{\bf{H}}^H}{{\bf{F}}^H} {\bf{FH}}} \right)}^{ - 1}}} \right]}_{k,k}}}}} \right)} \right\}}.
\end{equation}
First, we recall the Jensen's inequality on ${\log}_2 \left( 1+a \exp \left(x\right) \right)$ for $a>0$ \cite{Matthaiou2011}, which is expressed as
\begin{equation}\label{Jensen's inequality 1}
\mathbb{E} \left\{ {\log}_2 \left( 1 + a \exp \left(x\right) \right) \right\} \ge {\log}_2 \left( 1+a\exp\left( \mathbb{E} \left\{ x \right\}\right) \right).
\end{equation}
Applying \eqref{Jensen's inequality 1} into \eqref{ZF UL rate real2}, we can rewrite the uplink rate as
\begin{equation}\label{ZF UL rate appr1 proof1}
{R^{\rm ZF}} \ge \sum\limits_{k = 1}^{{N_u}} {{{\log }_2}\left( {1 + {P_{avg}}{\beta _k} \exp \left( \mathbb{E}\left\{{{X_k}}\right\} \right)} \right)} ,
\end{equation}
where
\begin{equation}\label{ZF UL rate appr1 proofX}
{X_k} = {\ln \left( {\frac{1}{{{{\left[ {{{\left( {{{\bf{H}}^H}{{\bf{F}}^H} {\bf{FH}}} \right)}^{ - 1}}} \right]}_{k,k}}}}} \right)} .
\end{equation}
For convenience, we denote ${{\bf{H}}_{eq}} = {\bf{FH}}$ as the equivalent instantaneous channel matrix. Then $X_k$ is equal to
\begin{equation}\label{ZF UL rate appr1 proofX1}
{X_k} = {\ln \left( {\frac{1}{{{{\left[ {{{\left( {{{\bf{H}}_{eq}^H} {\bf{H}}_{eq}} \right)}^{ - 1}}} \right]}_{k,k}}}}} \right)} .
\end{equation}
Since for any matrix ${\bf{Q}}$, it holds that \cite{Matthaiou2011}
\begin{equation}\label{ZF UL rate appr1 proof2}
{\left[ {{{\left( {{{\bf{Q}}^H}{\bf{Q}}} \right)}^{ - 1}}} \right]_{k,k}} = \frac{{\det \left( {{\bf{Q}}_k^H{{\bf{Q}}_k}} \right)}}{{\det \left( {{{\bf{Q}}^H}{\bf{Q}}} \right)}},
\end{equation}
where ${{\bf{Q}}_k}$ denotes ${\bf{Q}}$ with the $k$th column removed. Applying \eqref{ZF UL rate appr1 proofX1} and \eqref{ZF UL rate appr1 proof2} into \eqref{ZF UL rate appr1 proofX}, we write the expectation of $X_k$ as
\begin{equation}\label{ZF UL rate appr1 proof3}
\mathbb{E}\!\left\{{{X_k}}\right\}\! =\! \mathbb{E}\!\left\{{\ln \!\left({\det \left(\!{{\bf{H}}_{eq}^H{{\bf{H}}_{eq}}}\right)\!}\right)\!}\right\} - \mathbb{E}\!\left\{{\ln\!\left(\!{\det\left({{\bf{H}}_{eq,k}^H{{\bf{H}}_{eq,k}}}\right)\!}\right)\!}\right\}\!.
\end{equation}
To simplify the expressions, we define
\begin{equation}\label{ZF UL rate appr1 proofR}
{{\bf{R}}_{eq}} \triangleq {\bf{H}}_{eq}^H{{\bf{H}}_{eq}},\quad {{\bf{\bar R}}_{eq,k}}\triangleq{\bf{H}}_{eq,k}^H{{\bf{H}}_{eq,k}},
\end{equation}
where ${{\bf{H}}_{eq,k}}$ denotes ${{\bf{H}}_{eq}}$ with the $k$th column removed. Recalling \eqref{multiuser channel H}, the effective channel ${{\bf{H}}_{eq}}$ can be written as
\begin{equation}\label{ZF UL rate appr1 proof4}
{{\bf{H}}_{eq}} \!=\! {\bf{F\bar H}}{\left[ {{\bf{\Omega }}{{\left( {{\bf{\Omega }} \!+\! {{\bf{I}}_{{N_u}}}} \right)}^{ - 1}}} \right]^{\frac{1}{2}}} \!+\! {\bf{F}}{{\bf{H}}_w}{\left[ {{{\left( {{\bf{\Omega }} \!+\! {{\bf{I}}_{{N_u}}}} \right)}^{ - 1}}} \right]^{\frac{1}{2}}},
\end{equation}
where the first component is denoted by ${\bf{T}}$.
Since ${\bf{F}}{{\bf{H}}_w}$ is an extraction from the DFT transposition of ${{\bf{H}}_w}$, ${{\bf{H}}_{eq}}$ follows a Gaussian distribution with the mean matrix equal to ${\bf{T}}$ and the variance matrix of a row vector equal to ${\bf{\Sigma }} = {\left( {{\bf{\Omega }} + {{\bf{I}}_{{N_u}}}} \right)^{ - 1}}$. Hence, ${{\bf{R}}_{eq}}$ follows a non-central Wishart distribution, i.e., ${{\bf{R}}_{eq}}\sim {{\mathcal{W}}_{{N_u}}}\left( {{N_s},{\bf{T}},{\bf{\Sigma }}} \right)$. According to \cite{Steyn1972}, ${{\bf{R}}_{eq}}$ can be approximated by a central Wishart distribution with covariance ${\bf{\hat \Sigma }}$ defined in \eqref{central Wishart covariance}.
The positive definite Hermitian matrix ${\bf{\hat \Sigma }}$ can be eigenvalue decomposed by
\begin{equation}\label{ZF UL rate appr1 proof5}
\begin{aligned}
{\bf{\hat \Sigma }} &= {\bf{U}}_{\hat \Sigma} ^H{\bf{\Lambda }}{{\bf{U}}_{\hat \Sigma} },\\
{\bf{\Lambda }} &= {\rm diag}\left\{ {{\alpha _i}} \right\}_{i = 1}^{{N_u}},\\
\infty & \ge  {\alpha _{{N_u}}} \ge  \cdots  \ge {\alpha _1} \ge 0.
\end{aligned}
\end{equation}
Then we can further assume that
\begin{equation}\label{ZF UL rate appr1 proof R appr}
\det \left({{\bf{R}}_{eq}} \right) \approx \det \left({\bf{\Lambda }}{{\bf{H}}_1}{\bf{H}}_1^H \right),
\end{equation}
where ${{\bf{H}}_1}\in \mathbb{C}^{N_u \times N_s}$ follows a complex Gaussian distribution with $\bf{0}$ mean and ${{\bf{I}}_{{N_u}}} \otimes {{\bf{I}}_{{N_s}}}$  variance. Utilizing \emph{Lemma 4} of \cite{Jin2010}, we get that
\begin{equation}\label{ZF UL rate appr1 proof6}
\begin{aligned}
\mathbb{E}\left\{ {\ln \det \left( {{{\bf{R}}_{eq}}} \right)} \right\} &\approx \sum\limits_{i = 1}^{{N_u}} {\psi \left( {{N_s} - i + 1} \right)}  + \ln \det \left( {\bf{\Lambda }} \right)\\
 &= \sum\limits_{i = 1}^{{N_u}} {\left( {\psi \left( {{N_s} - i + 1} \right) + \ln {\alpha _i}} \right)} .
\end{aligned}
\end{equation}
Similarly, ${{\bf{\bar R}}_{eq,k}}$ satisfies $\det \left({{\bf{\bar R}}_{eq,k}}\right) \approx \det \left({{\bf{\bar \Lambda }}_k}{{\bf{H}}_{1,k}}{\bf{H}}_{1,k}^H \right)$, where ${\bf{H}}_{1,k}\in \mathbb{C}^{\left(N_u-1\right) \times N_s}$ follows a complex Gaussian distribution with $\bf{0}$ mean and ${{\bf{I}}_{{N_u-1}}} \otimes {{\bf{I}}_{{N_s}}}$ variance, ${{\bf{\bar \Lambda }}_k} = {\rm diag}\left\{ {{{\bar \alpha }_{k,i}}} \right\}_{i = 1}^{{N_u} - 1}$, and $\infty \ge {\bar \alpha _{k,{N_u} - 1}} \ge  \cdots  \ge {\bar \alpha _{k,1}} \ge 0$. Accordingly,
\begin{equation}\label{ZF UL rate appr1 proof7}
\mathbb{E}\left\{ {\ln \det \left( {{\bf{\bar R}}_{eq,k}} \right)} \right\} \!\approx\! \sum\limits_{i = 1}^{{N_u-1}} {\left( {\psi \left( {{N_s} \!-\! i \!+\! 1}\right) \!+\! \ln{\bar \alpha _{k,i}}} \right)} .
\end{equation}
With \eqref{ZF UL rate appr1 proof6} and \eqref{ZF UL rate appr1 proof7}, we can rewrite \eqref{ZF UL rate appr1 proof3} as
\begin{equation}\label{ZF UL rate appr1 proof8}
\begin{aligned}
\mathbb{E}\left\{{{X_k}}\right\} \approx &\sum\limits_{i = 1}^{{N_u}} {\left( {\psi \left( {{N_s} - i + 1} \right) + \ln {\alpha _i}} \right)} - \\
 &\sum\limits_{i = 1}^{{N_u} - 1} {\left( {\psi \left( {{N_s} - i + 1} \right) + \ln {{\bar \alpha }_{k,i}}} \right)}\\
 = & \psi \left( {{N_s} - {N_u} + 1} \right) + \sum\limits_{i = 1}^{{N_u}}{\ln {\alpha _i}} - \sum\limits_{i = 1}^{{N_u} - 1} { \ln {{\bar \alpha }_{k,i}}}\\
 = & \psi \left( {{N_s} - {N_u} + 1} \right) + \ln{\varepsilon _k}.
\end{aligned}
\end{equation}
Therefore, \eqref{ZF UL rate appr1} is obtained.

\subsection{Proof of Corollary \ref{ZF UL appr1 cor2}}\label{Apdx:TheoULZF1cor2}

When $K_k \to \infty $, the non-LoS (NLoS) components can be neglected. For ${\bf{\hat \Sigma }}$, it holds that
\begin{equation}\label{ZF UL rate cor2 proof1}
{\bf{\hat \Sigma }} \approx \frac{1}{N_s}{{\bf{\bar H}}^H}{{\bf{F}}^H} {\bf{F\bar H}}.
\end{equation}
The orthogonality among the equivalent LoS components contributes to
\begin{equation}\label{ZF UL rate cor2 proof2}
{{\bf{\bar H}}^H}{{\bf{F}}^H} {\bf{F\bar H}} = {\rm diag}\left\{{\left\| {{\bf{F\bar h}}_k} \right\|^2}\right\}_{k=1}^{N_u}.
\end{equation}
Hence, ${\bf{\hat \Sigma }}$ becomes diagonal and its eigenvalues correspond to the diagonal elements. Here, we assume ${\left\{ {{\alpha _i}} \right\}_{i = 1,\dots,{N_u}}}$ are unordered eigenvalues, it can be written that
\begin{equation}\label{ZF UL rate cor2 proof3}
\alpha _k = \frac{\left\| {{\bf{F\bar h}}_k} \right\|^2}{N_s}
\end{equation}
for $k=1,\dots, N_u$. Similarly, the eigenvalues of ${{{\bf{\hat \Sigma }}}_k}$ satisfy
\begin{equation}\label{ZF UL rate cor2 proof4}
{\bar \alpha }_{k,i} = \frac{\left\| {{\bf{F\bar h}}_j} \right\|^2}{N_s}
\end{equation}
for $i=1,\dots, N_u-1$ and $j=1,\dots,k-1,k+1,\dots,N_u-1$. As a consequence, it holds that
\begin{equation}\label{ZF UL rate cor2 proof5}
\sum\limits_{i = 1}^{{N_u}} {\ln {\alpha _i}}  - \sum\limits_{i = 1}^{{N_u} - 1} {\ln {{\bar \alpha }_{k,i}}}  = \ln \alpha _k.
\end{equation}
Applying \eqref{ZF UL rate cor2 proof4} and  \eqref{ZF UL rate cor2 proof5} into \eqref{ZF UL rate appr1} and \eqref{ZF UL rate appr1 epsk}, we can obtain \eqref{ZF UL rate cor2}.

\subsection{Proof of Theorem \ref{UL MRC appr}}\label{Apdx:TheoULMRC}
\begin{equation}\label{MRC UL rate appr proof1}
\begin{aligned}
&{R^{\rm{MRC}}_k} \approx\\
& {{{\log }_2}\left(\! {1 \!+\! \frac{{{P_{avg}}\mathbb{E}\! \left\{ {{{\left\| {{{\bf{g}}_{eq,k}}} \right\|}^4}} \right\}}}{{\sum\limits_{j \ne k} {P_{avg}}\mathbb{E}{\left\{ {{{\left| {{\bf{g}}_{eq,k}^H{{\bf{g}}_{eq,j}}} \right|}^2}} \right\}} \! +\! \mathbb{E}\!\left\{ {{{\left\| {{{\bf{g}}_{eq,k}}} \right\|}^2}} \right\}}}}\! \right)} .
\end{aligned}
\end{equation}
According to the definition of ${{\bf{G}}_{eq}}$, we can further write
\begin{equation}\label{MRC UL rate appr proof2}
{{\bf{g}}_{eq,k}} = {\bf{Fg}}_k = {\sqrt {\frac{\beta _k}{{{K_k} + 1}}} {\bf{Fh}}_{w,k} + \sqrt {\frac{{{K_k \beta _k}}}{{{K_k} + 1}}} {\bf{F\bar h}}_k} ,
\end{equation}
where ${{\bf{h}}_{w,k}}$ and ${{\bf{\bar h}}_k}$ are the $k$th column vector of ${{\bf{H}}_w}$ and ${\bf{\bar H}}$, respectively. Recalling the property of the DFT transformation of a complex Gaussian vector, we know that ${\bf{Fh}}_{w,k}$ is still an $N_s$ dimensional complex Gaussian vector. Therefore after derivations, we get the expression of the third expectation item in \eqref{MRC UL rate appr proof1} as
\begin{equation}\label{MRC UL rate appr x3 proof}
\begin{aligned}
\mathbb{E}\left\{ {{{\left\| {{{\bf{g}}_{eq,k}}} \right\|}^2}} \right\} &= \frac{{{\beta _k}}}{{{K_k} + 1}}\left( {{N_s} + {K_k}{{\left\| {{\bf{F\bar h}}_k^{}} \right\|}^2}} \right) \\
 &= \frac{{{\beta _k}}}{{{K_k} + 1}}{\chi _3^{(k)}},
\end{aligned}
\end{equation}
and the first and the second expectation items can be found in \eqref{MRC UL rate appr x1 proof} and \eqref{MRC UL rate appr x2 proof} respectively at the bottom of this page. Applying \eqref{MRC UL rate appr x3 proof}--\eqref{MRC UL rate appr x2 proof} into \eqref{MRC UL rate appr proof1}, we get the desired result.

\subsection{Proof of Theorem \ref{DL ZF1 appr}}\label{Apdx:TheoDLZF1}
Since $\rho$ is a constant during the coherence time of the channel, we can remove the expectation symbol in \eqref{ZF DL rate real} and rewrite it as
\setcounter{equation}{82}
\begin{equation}\label{ZF DL rate appr lt proof1}
{R^{\rm ZF1}} = \sum\limits_{k = 1}^{{N_u}} {{{\log }_2}\left( {1 + P{\rho ^2}} \right)}.
\end{equation}
Then we turn to the calculation for the expression of $\rho ^2$. After a few steps of matrix transformation, it can be derived that
\begin{equation}\label{ZF DL rate appr lt proof2}
\begin{aligned}
\mathbb{E}\left\{ {\left\| {\bf{\bar W}} \right\|_F^2} \right\} &=\mathbb{E}\left\{ {{\rm{trace}}\left( {{\bf{\bar W}}^H{\bf{\bar W}}} \right)} \right\}\\
 &= \mathbb{E}\left\{ {{\rm{trace}}\left( {{{\left( {{\bf{G}}_{eq}^T {{\bf{G}}_{eq}^{*}}} \right)}^{ - 1}}} \right)} \right\}\\
 &= \sum\limits_{k = 1}^{{N_u}} {\beta _k^{ - 1} \mathbb{E}\left\{{\left[ {{{\bf{R}}_{eq}^{ - 1}}} \right]_{k,k}^{*}}\right\}}.
\end{aligned}
\end{equation}
Since ${\bf{\hat \Sigma }}$ is symmetric and Hermitian with positive main diagonal elements, its inverse matrix still holds the same characteristics, that is, ${{\left[ {{{{\bf{\hat \Sigma }}}^{ - 1}}} \right]}_{k,k}^{*}} ={{\left[ {{{{\bf{\hat \Sigma }}}^{ - 1}}} \right]}_{k,k}}$ for $k = 1,\dots,N_u$. Then, we recall the central Wishart approximation of ${{\bf{R}}_{eq}} \approx {\bf{\Lambda}} {{\bf{H}}_1}{\bf{H}}_1^H$. Utilizing \emph{Theorem} of \cite{Gore2002}, we know that ${\gamma _k} = {1 \mathord{\left/ {\vphantom {1 {\left( {{{\bf{R}}_{eq}}} \right)_{kk}^{ - 1}}}} \right. \kern-\nulldelimiterspace} {\left[ {{{\bf{R}}_{eq}^{ - 1}}} \right]_{k,k}}}$ satisfies the Chi-squared distribution
\begin{equation}\label{ZF DL rate appr lt proof3}
f\left( {{\gamma _k}} \right) = \frac{{{{\left[ {{{{\bf{\hat \Sigma }}}^{ - 1}}} \right]}_{k,k}}{e^{ - {\gamma _k}{{\left[ {{{{\bf{\hat \Sigma }}}^{ - 1}}} \right]}_{k,k}}}}}}{{\left( {{N_s} - {N_u}} \right)!}}{\left( {{\gamma _k}{{\left[ {{{{\bf{\hat \Sigma }}}^{ - 1}}} \right]}_{k,k}}} \right)^{{N_s} - {N_u}}},
\end{equation}
and the expectation of ${\gamma _k^{-1}}$ is
\begin{equation}\label{ZF DL rate appr lt proof4}
\mathbb{E}\left\{ {\gamma _k^{-1}} \right\} = \frac{{{{\left[ {{{{\bf{\hat \Sigma }}}^{ - 1}}} \right]}_{k,k}}}}{{{N_s} - {N_u}}}.
\end{equation}
Substituting
\begin{equation}\label{ZF DL rate appr lt proof5}
{\rho ^2} = \frac{1}{\sum\limits_{k = 1}^{{N_u}} {\beta _k^{ - 1}\mathbb{E}\left\{{\gamma _k^{(-1)*}} \right\}}}= \frac{{\left( {{N_s} - {N_u}} \right)}}{{\sum\limits_{k = 1}^{{N_u}} {\beta _k^{ - 1}{{\left[ {{{{\bf{\hat \Sigma }}}^{ - 1}}} \right]}_{k,k}}} }}
\end{equation}
into \eqref{ZF DL rate appr lt proof1}, we get the approximation \eqref{ZF DL rate appr lt}.

\subsection{Proof of Theorem \ref{DL ZF2 appr}}\label{Apdx:TheoDLZF2}
Under this condition, $\rho _1, \dots, \rho_{N_u} $ are not constant any more. Then the achievable rate \eqref{ZF DL rate real} can be approximated by
\begin{equation}\label{ZF DL rate appr st proof1}
{R^{\rm ZF2}} \approx \sum\limits_{k = 1}^{{N_u}} {{{\log }_2}\left( {1 + P \mathbb{E}\left\{ {\rho _k^2} \right\}} \right)} .
\end{equation}
Recalling the definition of $\rho _k$, we know
\begin{equation}\label{ZF DL rate appr st proof2}
{\rho _k^2} = \frac{1}{{N_u}{\left\| {\bf{\bar w}}_k \right\|}^2}.
\end{equation}
Utilizing the matrix transformation property, ${\left\| {\bf{\bar w}}_k \right\|}^2$ can be calculated as
\begin{equation}\label{ZF DL rate appr st proof3}
{\left\| {\bf{\bar w}}_k \right\|}^2 = {\left[ {{\bf{\bar W}}{\bf{\bar W}}^H} \right]}_{k,k}
={\beta_k^{-1}} \left[ {\bf{R}}_{eq} ^{ - 1}\right]_{k,k}^{*} = {\beta_k^{-1}} {\gamma _k^{(-1)*}}.
\end{equation}
According to \eqref{ZF DL rate appr lt proof3}, we can derive that
\begin{equation}\label{ZF DL rate appr st proof4}
\mathbb{E}\left\{ {\rho _k^2} \right\} = \frac{{{\beta _k}}}{{{N_u}}}\mathbb{E}\left\{ {{\gamma _k^{*}}} \right\} = \frac{{{\beta _k}\left( {{N_s} - {N_u} + 1} \right)}}{{{N_u}{{\left[ {{{{\bf{\hat \Sigma }}}^{ - 1}}} \right]}_{k,k}}}}.
\end{equation}
Therefore, applying \eqref{ZF DL rate appr st proof4} into \eqref{ZF DL rate appr st proof1}, \eqref{ZF DL rate appr st} is formulated.

\subsection{Proof of Theorem \ref{DL MRT1 appr}}\label{Apdx:TheoDLMRT1}
According to \emph{Lemma 1} of \cite{Zhang2014}, \eqref{MRT DL rate real} can be approximated by
\begin{equation}\label{MRT DL rate appr lt proof1}
{R^{\rm{MRT1}}} \approx \sum\limits_{k = 1}^{{N_u}} {{{\log }_2}\left( {1 + \frac{{P\rho _k^2 \mathbb{E}\left\{ {{{\left\| {{{\bf{g}}_{eq,k}}} \right\|}^4}} \right\}}}{{\sum\limits_{j \ne k} {P\rho _j^2 \mathbb{E}\left\{ {{{\left| {{{\bf{g}}_{eq,k}^H}{\bf{g}}_{eq,j}} \right|}^2}} \right\}}  + 1}}} \right)},
\end{equation}
where the expressions of the two expectation items can be found in \eqref{MRC UL rate appr x1 proof} and \eqref{MRC UL rate appr x2 proof}, respectively. When it comes to $\rho _k^2$, we first write the definition of long-term normalization of the MRT precoder as
\begin{equation}\label{MRT long-term nomalization}
\rho  = {\rho _1} = {\rho _2} =  \cdots  = {\rho _{{N_u}}} = \frac{1} {\sqrt{\mathbb{E}\left\{ {\left\| {\bf{G}}_{eq} \right\|_F^2} \right\}} }.
\end{equation}
Then we can derive that
\begin{equation}\label{MRT DL rate appr lt proof2}
\begin{aligned}
\mathbb{E}{\left\{ {\left\| {{{\bf{G}}_{eq}}} \right\|_F^2} \right\}}
&= \mathbb{E}{\left\{ {\sum\limits_{i = 1}^{{N_u}} {{{\left\| {{{\bf{g}}_{eq,i}}} \right\|}^2}} } \right\}}\\
&= \sum\limits_{i = 1}^{{N_u}} {\frac{{{\beta _i}}}{{{K_i} + 1}}\left( {{N_s} + {K_i}{{\left\| {\bf{F}}{{\bf{\bar h}}_i} \right\|}^2}} \right)} \\
&= \sum\limits_{i = 1}^{{N_u}} {\frac{{{\beta _i}}}{{{K_i} + 1}}\chi _3^{(i)}} .
\end{aligned}
\end{equation}
Applying \eqref{MRC UL rate appr x1 proof}, \eqref{MRC UL rate appr x2 proof} and \eqref{MRT DL rate appr lt proof2} into \eqref{MRT DL rate appr lt proof1}, it returns \eqref{MRT DL rate appr lt}.

\subsection{Proof of Theorem \ref{DL MRT2 appr}}\label{Apdx:TheoDLMRT2}
According to the definition of the short-term normalization of the MRT precoders, i.e.,
\begin{equation}\label{MRT short-term nomalization}
{\rho _k} = \frac{1}{{\sqrt {{N_u}} \left\| {\bf{g}}_{eq,k} \right\|}},
\end{equation}
the achievable rate satisfies
\begin{equation}\label{MRT DL rate appr st proof1}
{R^{\rm{MRT2}}} = \sum\limits_{k = 1}^{{N_u}}\! \mathbb{E}{\left\{ {{{\log }_2}\left( {1 \!+\! \frac{{{\frac{P}{N_u}} {{\left\| {{{\bf{g}}_{eq,k}}} \right\|}^2}}}{{\sum\limits_{j \ne k} {\frac{{{P{\left| {{{\bf{g}}_{eq,k}^H}{\bf{g}}_{eq,j}} \right|}^2}}}{{{N_u{\left\| {{{\bf{g}}_{eq,j}}} \right\|}^2}}}} \!+\! {\rm{1}}}}} \right)} \right\}} .
\end{equation}
Recalling \emph{Lemma 1} of \cite{Zhang2014}, \eqref{MRT DL rate appr st proof1} can be approximated by
\begin{equation}\label{MRT DL rate appr st proof2}
{R^{\rm{MRT2}}}\approx \sum\limits_{k = 1}^{{N_u}} {{{\log }_2}\left( {1 + \frac{{{\frac{P}{N_u}} \mathbb{E}\left\{ {{{\left\| {{{\bf{g}}_{eq,k}}} \right\|}^2}} \right\}}}{{\sum\limits_{j \ne k} {\frac{P {\mathbb{E}\left\{ {{{\left| {{{\bf{g}}_{eq,k}^H}{\bf{g}}_{eq,j}} \right|}^2}} \right\}}}{N_u{\mathbb{E}\left\{ {{{\left\| {{{\bf{g}}_{eq,j}}} \right\|}^2}} \right\}}}}  + {\rm{1}}}}} \right)}.
\end{equation}
Utilizing the results in \eqref{MRC UL rate appr x3 proof} and \eqref{MRC UL rate appr x2 proof}, we can obtain \eqref{MRT DL rate appr st}.

\end{document}